\documentclass[journal, onecolumn]{IEEEtran}
\usepackage{cite}
\usepackage{amssymb}
\usepackage[cmex10]{amsmath,nccmath}
\usepackage{mathtools}
\usepackage{bm}
\usepackage{orcidlink}
\usepackage{hyperref}
\usepackage{url}
\hypersetup{colorlinks=true,linkcolor=blue,citecolor=blue,urlcolor=blue}

\newtheorem{theorem}{Theorem}
\newtheorem{lemma}{Lemma}
\newtheorem{definition}{Definition}
\newtheorem{corollary}{Corollary}

\newcommand{\E}[1]{\mathbb E\left[#1\right]}
\newcommand{\Prob}[1]{\mathbb P\left[#1\right]}
\newcommand{\mc}[1]{\mathcal{#1}}
\newcommand{\ms}[1]{\mathsf{#1}}
\newcommand{\mr}[1]{\mathrm{#1}}
\newcommand{\bs}[1]{\boldsymbol{#1}}
\newcommand{\Pe}{P_{\mr{e}}}
\newcommand{\Mvlf}{M^*}
\newcommand{\Nvlf}{N^*}
\newcommand{\thmref}[1]{Theorem~\ref{#1}}
\newcommand{\lemref}[1]{Lemma~\ref{#1}}

\newcommand{\secref}[1]{Section~\ref{#1}}
\newcommand{\appref}[1]{Appendix~\ref{#1}}
\newcommand{\defnref}[1]{Definition~\ref{#1}}

\allowdisplaybreaks

\begin{document}

\title{A Tight Second-Order Converse Bound for Variable-Length Feedback Codes}
\author{Recep Can Yavas~\orcidlink{0000-0002-5640-515X},~\IEEEmembership{Member,~IEEE}
\thanks{R. C. Yavas is with the Department of Electrical and Electronics Engineering at Bilkent University, 06800, Ankara, Turkey (e-mail: ryavas@bilkent.edu.tr).}}
\IEEEoverridecommandlockouts
\maketitle
\bstctlcite{BSTcontrol}

\begin{abstract}
We study variable-length feedback (VLF) codes over a discrete memoryless channel under average decoding-time and error-probability constraints. In the non-vanishing error probability regime, Polyanskiy, Poor, and Verd\'u (2011) derive achievability and converse bounds on the logarithm of the maximum achievable codebook size. These bounds establish the $\epsilon$-capacity but leave an order-$\log N$ gap in the second-order expansion, where $N$ is the average decoding time. Yavas and Tan (2025) improve the coefficient of $\log N$ in the achievability bound from $-1$ to $-\frac{C}{C_1}$, where $C$ is the channel capacity and $C_1$ is the largest Kullback--Leibler divergence between two conditional output distributions. We derive a converse with the same coefficient, establishing the second-order fundamental limit for every positive-capacity discrete memoryless channel with finite $C_1$. The result also covers the moderate-deviations and error-exponent regimes, including polynomially decaying error probabilities. The converse uses R\'enyi entropy and the extrinsic Jensen--Shannon divergence. We also derive necessary properties of asymptotically optimal VLF codes. First-order-optimal codes must have an early-stopping branch, and second-order-optimal codes must additionally exhibit communication and confirmation behavior. Finally, for the binary erasure channel, we determine the exact minimum expected decoding time for every message-set size and admissible error probability.
\end{abstract}


\section{Introduction} \label{sec:introduction}

A variable-length feedback (VLF) code allows the encoder to select each channel input using the preceding channel outputs and allows the decoder to terminate at a stopping time adapted to those outputs. The decoding time therefore responds to the realized noise, while performance is measured through its expectation. Output-adapted stopping improves the finite-length performance of communication over a discrete memoryless channel (DMC)~\cite{PolyanskiyPoorVerdu2011}, while causal input adaptation together with output-adapted stopping yields Burnashev's optimal VLF error exponent~\cite{Burnashev1976}, which is in general larger than that of fixed-length codes without feedback. 

Let $C$ denote the channel capacity and let $C_1$ denote the largest Kullback--Leibler (KL) divergence between two conditional output distributions. Fix a DMC with $C>0$ and $C_1<\infty$ and a rate $R\in(0,C)$. Consider a sequence of VLF codes indexed by $j$, with $M_j$ equiprobable messages, average error probabilities $\epsilon_j\in(0,1)$, and decoding times $\tau_j$, such that $\epsilon_j\to0$, $\E{\tau_j}\to\infty$, and $\frac{\log M_j}{\E{\tau_j}}\to R$. The error exponent of this sequence is $\liminf_{j\to\infty}\frac{\log\frac{1}{\epsilon_j}}{\E{\tau_j}}$. In~\cite[Th.~3]{Burnashev1976}, Burnashev proves that the supremum of these exponents over all such code sequences is
\begin{align}
    E(R) &= C_1\left(1-\frac{R}{C}\right), \qquad R\in(0,C). \label{eq:intro-burnashev}
\end{align}

Yamamoto and Itoh~\cite{YamamotoItoh1979} attain Burnashev's exponent by alternating fixed-length communication and confirmation phases. Each communication phase uses a fixed-length block code operated at a rate approaching capacity to form a tentative decision. In the following confirmation phase, the transmitter sends an accept or reject control sequence formed from the two input symbols that attain $C_1$, and the decoder performs a fixed-length binary hypothesis test. An accepted decision terminates transmission, whereas a rejected decision initiates another communication phase. Berlin \emph{et al.}~\cite{BerlinEtAl2009} derive a converse bound that parallels this two-phase construction and suggests that communication and confirmation phases are implicit in schemes attaining Burnashev's exponent.

Polyanskiy \emph{et al.}~\cite{PolyanskiyPoorVerdu2011} study VLF coding in the non-vanishing error probability regime. For fixed $\epsilon\in(0,1)$, let $\Mvlf(N,\epsilon)$ denote the largest codebook size achievable with average decoding time at most $N$ and average error probability at most $\epsilon$. Their achievability bound in~\cite[Th.~2]{PolyanskiyPoorVerdu2011} and converse bound in~\cite[Th.~4]{PolyanskiyPoorVerdu2011} give
\begin{align}
    \frac{NC}{1-\epsilon}-\log N+O(1) &\leq \log \Mvlf(N,\epsilon) \leq \frac{NC+h_{\mr{b}}(\epsilon)}{1-\epsilon}, \label{eq:intro-ppv}
\end{align}
where $h_{\mr{b}}(\cdot)$ is the binary Shannon entropy function. The lower bound is attained by stop-feedback codes.

Yavas and Tan~\cite{YavasTan2025} sharpen the achievability bound for every DMC with $C>0$ and $C_1<\infty$. Their modified Yamamoto--Itoh scheme has two communication phases separated by one confirmation phase, and each phase ends at a random stopping time. For every fixed $\epsilon\in(0,1)$, their non-asymptotic construction in~\cite[Th.~1]{YavasTan2025}, after asymptotic optimization in~\cite[Th.~2]{YavasTan2025}, gives
\begin{align}
    \log \Mvlf(N,\epsilon) &\geq \frac{NC}{1-\epsilon} -\frac{C}{C_1}\log N -\log\log N +O(1). \label{eq:intro-yavas-tan}
\end{align}

Yavas~\cite{Yavas2026} analyzes posterior matching over binary-input memoryless symmetric (BMS) channels. For fixed $\epsilon\in(0,\frac{1}{2})$, the resulting construction improves the third-order term $-\log \log N + O(1)$ in~\eqref{eq:intro-yavas-tan} to $O(1)$ over BMS DMCs with positive capacity and finite $C_1$~\cite[Cor.~1]{Yavas2026}. Related posterior-based work gives channel-specific achievability results for binary asymmetric channels~\cite{YangEtAl2022} and develops low-complexity posterior matching over the binary symmetric channel~\cite{AntoniniEtAl2024}.

In the moderate-deviations regime, the coding rate approaches capacity while the error probability vanishes. Truong and Tan~\cite{TruongTan2019} consider the special case in which $\epsilon_N\to0$ and $\frac{1}{N}\log M_N\geq C-\rho_N$, where $\rho_N\to0$ and $\sqrt{N}\rho_N\to\infty$. For DMCs with finite $C_1$, they prove in~\cite[Th.~1]{TruongTan2019} that the largest achievable value of $\liminf_{N\to\infty}\frac{\log\frac{1}{\epsilon_N}}{N\rho_N}$ is $\frac{C_1}{C}$. Since $\sqrt{N}\rho_N\to\infty$ implies $\frac{N\rho_N}{\log N}\to\infty$, the error probabilities at this optimal scale decay faster than every inverse polynomial in $N$. Thus, their result does not cover polynomially decaying error probabilities.

Earlier VLF converses use different measures of the decoder's uncertainty. Burnashev bounds two posterior drifts by $C$ and $C_1$ and combines these bounds with a terminal entropy estimate~\cite{Burnashev1976}. Berlin \emph{et al.}~\cite{BerlinEtAl2009} split their converse at the first time the conditional maximum a posteriori (MAP) error probability falls below a threshold and lower-bound the expected times before and after that crossing using capacity and binary hypothesis testing. The converse behind the upper bound in~\eqref{eq:intro-ppv} combines Fano's inequality, a stopped mutual-information decomposition, and the one-step capacity bound.

For DMCs with $C>0$ and $C_1<\infty$, the bounds in~\eqref{eq:intro-ppv} and~\eqref{eq:intro-yavas-tan} do not determine whether the coefficient of $\log N$ in the VLF fundamental limit is $-\frac{C}{C_1}$. They also do not give a single expansion that determines the transition between the $\log N$ and $\log\frac{1}{\epsilon_N}$ scales.

\subsection{Main Contributions}

For every DMC with $C>0$ and $C_1<\infty$, our main result closes the order-$\log N$ converse gap in the non-vanishing error probability regime and determines how the $\log N$ and $\log\frac{1}{\epsilon_N}$ scales enter the VLF fundamental limit. For error sequences bounded away from one and satisfying $\log\frac{1}{\epsilon_N}\leq(C_1-\eta)N$ for some $\eta\in(0,C_1)$, we prove
\begin{align}
    \log \Mvlf(N,\epsilon_N) &= \frac{NC}{1-\epsilon_N} -\frac{C}{C_1} \max\left\{\log N,\log\frac{1}{\epsilon_N}\right\} +o\left(\max\left\{\log N,\log\frac{1}{\epsilon_N}\right\}\right). \label{eq:intro-main}
\end{align}
In the non-vanishing error probability regime, the converse proves that the coefficient of $\log N$ is $-\frac{C}{C_1}$ throughout this channel class, matching the achievability bound in~\eqref{eq:intro-yavas-tan}. The same expansion connects the non-vanishing error probability, moderate-deviations, and error-exponent regimes and includes polynomially decaying error probabilities. The achievability part reoptimizes the non-asymptotic construction in~\cite[Th.~1]{YavasTan2025}, with a parameter choice uniform over the permitted error-probability sequences.

In a variable-length stop-feedback (VLSF) code, the channel-input sequence is fixed by the message and common randomness and does not depend on past channel outputs. Feedback only tells the encoder when decoding occurs. Hence, VLSF codes form a subclass of VLF codes, and every VLF converse also applies to them. In~\cite[Th.~9]{YavasKostinaEffros2024Sparse}, Yavas \emph{et al.} derive a non-asymptotic VLSF converse in terms of the minimum type-II error probability of a sequential hypothesis test. When decoding is allowed after every channel use, the asymptotic analysis of that bound gives only $\log M\leq\frac{NC}{1-\epsilon}+O(1)$. On the other hand, for DMCs with $C>0$ and $C_1<\infty$, the converse underlying~\eqref{eq:intro-main} improves this bound to $\log M\leq\frac{NC}{1-\epsilon}-\frac{C}{C_1}\log N+o(\log N)$. To the best of our knowledge, the lower bound in~\eqref{eq:intro-ppv} remains the best general VLSF achievability bound, and its coefficient of $\log N$ is $-1$. Since $C<C_1$ for every DMC with $C>0$, the achievability coefficient $-1$ differs from the converse coefficient $-\frac{C}{C_1}$. Closing this gap remains an open problem.

The converse has two stages. The first supplements Burnashev's Shannon-entropy argument in~\cite[Lem.~2]{Burnashev1976} with a R\'enyi-entropy bound of order $1+\frac{2}{\log(M-1)}$. The order-$\alpha$ Sibson capacity~\cite{Sibson1969} upper-bounds the expected one-step decrease of the posterior R\'enyi entropy (see~\lemref{lem:renyi}). At this order, the accumulated excess of the Sibson-capacity bound over $C$ contributes $O(1)$ to the upper bound on $\log M$, while the terminal R\'enyi-entropy bound contains a negative term of order $\log(M-1)$ times the product of the conditional MAP error probability and its complement (see~\lemref{lem:terminal-renyi}). Comparing the R\'enyi and Shannon bounds then shows that the terminal conditional MAP error of a near-optimal code sequence must concentrate near zero or one, yielding the logarithmic lower bound on the terminal posterior potential used in the second stage.

The second stage relates this terminal potential to the accumulated capacity deficit. In~\cite[Lem.~2]{NaghshvarJavidiWigger2015}, Naghshvar \emph{et al.} identify the extrinsic Jensen--Shannon (EJS) divergence with the expected one-step increase of the posterior average log-odds. Together with the Shannon-entropy drift, this identity expresses the expected one-step increase of our potential as the EJS divergence minus the mutual information (see~\eqref{eq:support-L-drift}). We upper-bound this difference by a multiple of the one-step capacity deficit plus a remainder whose expected stopped sum is bounded. Summing this bound until decoding yields the logarithmic term in~\eqref{eq:intro-main}.

Our first complementary result gives necessary conditions for first- and second-order-optimal VLF code sequences in the non-vanishing error probability regime. In every first-order-optimal sequence, an event whose probability approaches $\epsilon$ accounts for almost all decoding errors and has conditional mean decoding time $o(N)$. On the complementary event, decoding is asymptotically reliable and the conditional mean decoding time is $\frac{N}{1-\epsilon}+o(N)$. Thus, errors arise almost entirely from a branch that uses negligible time on average, whereas decoding is reliable on the other branch. Second-order optimality further requires communication and confirmation behavior, although these behaviors need not occur in two contiguous phases.

The second complementary result treats the binary erasure channel (BEC), for which $C_1=\infty$. For every integer $M\geq2$, erasure probability $\delta\in[0,1)$, and $\epsilon\in[0,1-\frac{1}{M}]$, we determine the \emph{exact minimum expected decoding time}. At $\epsilon=0$, this exact value agrees with~\cite[Cor.~6]{DevassyEtAl2016}; when, in addition, $M=2^k$ for an integer $k\geq1$, it equals the expected decoding time achieved by the construction in~\cite[Th.~7]{PolyanskiyPoorVerdu2011}. For $0<\epsilon\leq1-\frac{1}{M}$, the exact expression strengthens the earlier non-asymptotic bounds of Devassy \emph{et al.}~\cite{DevassyEtAl2016}.

\subsection{Paper Organization}

\secref{sec:notation} introduces the notation and channel quantities, and \secref{sec:problem} defines the VLF code model and its fundamental limit. \secref{sec:main} states the main results, and \secref{sec:supporting} develops the posterior inequalities used in the converse. The next three sections contain the proofs, \secref{sec:conclusion} concludes the paper, and the appendices contain the uniform achievability calculation and the auxiliary proofs.

\section{Notation and Definitions} \label{sec:notation}

For a positive integer $m$, let $[m]\triangleq\{1,\ldots,m\}$. The sets of positive and non-negative integers are $\mathbb N\triangleq\{1,2,\ldots\}$ and $\mathbb N_0\triangleq\{0,1,\ldots\}$, respectively. For a sequence $x_1,x_2,\ldots$, let $x^n\triangleq(x_1,\ldots,x_n)$, with $x^0$ denoting the empty sequence. We write $1\{\cdot\}$ for the indicator function. For events $\mc{A}$ and $\mc{B}$, the notation $\mc{A}^{\mr{c}}$ denotes the complement of $\mc{A}$, and $\mc{A}\mathbin{\triangle}\mc{B}$ denotes their symmetric difference. All unsubscripted logarithms have base $e$. For $a\in\mathbb R$, let $a^+\triangleq\max\{a,0\}$. We use the standard asymptotic notation $o(\cdot)$, $O(\cdot)$, $\Omega(\cdot)$, and $\Theta(\cdot)$.

For a probability vector $\bs{p}=(p_1,\ldots,p_m)$, its Shannon entropy and its R\'enyi entropy of order $\alpha>1$ are
\begin{align}
    H(\bs{p}) &\triangleq -\sum_{i=1}^{m}p_i\log p_i, \label{eq:notation-shannon-entropy}\\
    H_\alpha(\bs{p}) &\triangleq \frac{1}{1-\alpha} \log\sum_{i=1}^{m}p_i^\alpha. \label{eq:notation-renyi-entropy}
\end{align}
Here and below, $0\log0\triangleq0$. The binary Shannon entropy in nats is
\begin{align}
    h_{\mr{b}}(a) &\triangleq -a\log a-(1-a)\log(1-a), \qquad a\in[0,1]. \label{eq:notation-binary-entropy}
\end{align}
For discrete random variables, $H(\cdot)$, $H(\cdot\mid\cdot)$, $I(\cdot;\cdot)$, and $I(\cdot;\cdot\mid\cdot)$ denote the Shannon entropy, the conditional entropy, the mutual information, and the conditional mutual information, respectively.

For probability distributions $P$ and $Q$ on the same finite alphabet such that $P$ is absolutely continuous with respect to $Q$, written $P\ll Q$, the KL divergence and the chi-square divergence are
\begin{align}
    D(P\|Q) &\triangleq \sum_{y:P(y)>0}P(y)\log\frac{P(y)}{Q(y)}, \label{eq:notation-kl}\\
    \chi^2(P\|Q) &\triangleq \sum_{y:Q(y)>0}\frac{(P(y)-Q(y))^2}{Q(y)}. \label{eq:notation-chi-square}
\end{align}
If $P$ is not absolutely continuous with respect to $Q$, we set both divergences equal to infinity.

For a probability vector $\bs{p}=(p_1,\ldots,p_m)$ with $p_i<1$ for every $i\in[m]$ and probability distributions $P_1,\ldots,P_m$ on a common finite alphabet, the EJS divergence in~\cite[Eq.~(9a)]{NaghshvarJavidiWigger2015} is defined as
\begin{align}
    \operatorname{EJS}(\bs{p};P_1,\ldots,P_m) &\triangleq \sum_{i:p_i>0}p_iD\left(P_i\middle\|\sum_{j\neq i}\frac{p_j}{1-p_i}P_j\right), \label{eq:notation-ejs}
\end{align}
where indices with zero weight are omitted.

Let $P_{Y|X}$ be a DMC with input alphabet $\mc{X}$, output alphabet $\mc{Y}$, and transition probabilities $P_{Y|X}(y|x)$. For every $n\in\mathbb N$, $x^n\in\mc{X}^n$, and $y^n\in\mc{Y}^n$, its memoryless extension satisfies
\begin{align}
    P_{Y^n|X^n}(y^n|x^n) &= \prod_{t=1}^{n}P_{Y|X}(y_t|x_t). \label{eq:notation-memoryless}
\end{align}

For an input distribution $P_X$ on $\mc{X}$, let $I(X;Y)$ denote the mutual information under $P_XP_{Y|X}$. The channel capacity and the largest KL divergence between two conditional output distributions are
\begin{align}
    C &\triangleq \max_{P_X}I(X;Y), \label{eq:notation-capacity}\\
    C_1 &\triangleq \max_{x,x'\in\mc{X}} D\left(P_{Y|X=x}\middle\|P_{Y|X=x'}\right). \label{eq:notation-C1}
\end{align}
For $C_1<\infty$, convexity of the KL divergence in its second argument gives $I(X;Y)\leq C_1\left(1-\sum_xP_X(x)^2\right)$ for every $P_X$, which implies that $C\leq C_1$. If $C>0$ and $C_1\in(0,\infty)$, evaluating the same inequality at a capacity-achieving input distribution gives $C<C_1$. When $C_1=\infty$, the strict inequality is immediate. Thus, every positive-capacity DMC satisfies $C<C_1$.

\section{Problem Formulation} \label{sec:problem}

The following definition uses the VLF code model in~\cite[Def.~1]{PolyanskiyPoorVerdu2011}.

\begin{definition}[VLF code]\label{def:vlf-code}
Fix $N>0$, a positive integer $M$, and $\epsilon\in[0,1]$. An $(N,M,\epsilon)$-VLF code consists of the following objects.
\begin{enumerate}
\item A probability distribution $P_U$ on an alphabet $\mc{U}$ with $|\mc{U}|\leq3$. The common randomness $U\sim P_U$ is revealed to the encoder and decoder before transmission and is independent of the message $W$, which is equiprobable on $[M]$.

\item A sequence of causal encoding functions
\begin{align}
    \ms{f}_t &\colon \mc{U}\times[M]\times\mc{Y}^{t-1} \to \mc{X}, \qquad t\in\mathbb N, \label{eq:problem-encoder-map}
\end{align}
which generate the channel inputs
\begin{align}
    X_t &= \ms{f}_t(U,W,Y^{t-1}), \qquad t\in\mathbb N. \label{eq:problem-encoder}
\end{align}
For every $t\in\mathbb N$, conditioned on $(U,W,Y^{t-1})$, the channel output $Y_t$ has distribution $P_{Y|X}\left(\cdot\middle|X_t\right)$.

\item A sequence of decoding functions
\begin{align}
    \ms{g}_t &\colon \mc{U}\times\mc{Y}^t \to [M], \qquad t\in\mathbb N_0. \label{eq:problem-decoder-map}
\end{align}

\item A non-negative integer-valued stopping time $\tau$ with respect to the filtration
\begin{align}
    \mc{F}_t &\triangleq \sigma(U,Y^t), \qquad t\in\mathbb N_0. \label{eq:problem-filtration}
\end{align}
On the event $\{\tau=t\}$, the decoder produces the terminal estimate
\begin{align}
    \widehat W &= \ms{g}_t(U,Y^t), \qquad t\in\mathbb N_0. \label{eq:problem-terminal-estimate}
\end{align}
The average decoding time and average error probability satisfy
\begin{align}
    \E{\tau} &\leq N, \label{eq:problem-length-constraint}\\
    \Pe \triangleq \Prob{\widehat W\neq W} &\leq \epsilon. \label{eq:problem-error-constraint}
\end{align}
\end{enumerate}
\end{definition}

For $N>0$ and $\epsilon\in[0,1)$, the VLF fundamental limit is the largest admissible message-set size,
\begin{align}
    \Mvlf(N,\epsilon) &\triangleq \max\left\{ M\in\mathbb N: \text{an $(N,M,\epsilon)$-VLF code exists} \right\}. \label{eq:problem-vlf-fundamental-limit}
\end{align}

\section{Main Results} \label{sec:main}

\subsection{A Tight Converse Bound for VLF Codes}

We first state the asymptotic fundamental limit for an error-probability sequence that may vary with $N$.

\begin{theorem}\label{thm:main}
Consider a DMC with $C>0$ and $C_1<\infty$. Fix $\bar\epsilon\in(0,1)$ and $\eta\in(0,C_1)$, and let $\{\epsilon_N\}_{N\in\mathbb N}$ be a positive sequence such that, for all sufficiently large $N$, $\epsilon_N\leq\bar\epsilon$ and
\begin{align}
    \log\frac{1}{\epsilon_N} &\leq (C_1-\eta)N. \label{eq:main-reliability-condition}
\end{align}
Define
\begin{align}
    A_N &\triangleq \max\left\{\log N,\log\frac{1}{\epsilon_N}\right\}. \label{eq:main-A-N}
\end{align}
Then, as $N\to\infty$,
\begin{align}
    \log \Mvlf(N,\epsilon_N) &\geq \frac{NC}{1-\epsilon_N}-\frac{C}{C_1}A_N-\left(1-\frac{C}{C_1}\right)\log A_N-O(1), \label{eq:main-achievability}\\
    \log \Mvlf(N,\epsilon_N) &\leq \frac{NC+h_{\mr{b}}(\epsilon_N)}{1-\epsilon_N}-\frac{C}{C_1}A_N+o(A_N). \label{eq:main-converse}
\end{align}
Consequently,
\begin{align}
    \log \Mvlf(N,\epsilon_N) &= \frac{NC}{1-\epsilon_N}-\frac{C}{C_1}A_N+o(A_N). \label{eq:main-expansion}
\end{align}
The $O(1)$ term in~\eqref{eq:main-achievability} depends only on the channel and $\bar\epsilon$. The $o(A_N)$ terms in~\eqref{eq:main-converse} and~\eqref{eq:main-expansion} are uniform over all pairs $(N,\epsilon_N)$ satisfying $0<\epsilon_N\leq\bar\epsilon$ and~\eqref{eq:main-reliability-condition}.
\end{theorem}

\begin{IEEEproof}
See \appref{app:achievability} for the achievability bound and \secref{sec:proof-main} for the converse.
\end{IEEEproof}

The converse in~\eqref{eq:main-converse} is the principal result. 
Three specializations describe the non-vanishing error probability regime, polynomially decaying error probabilities, and exponential error decay.

\begin{corollary}\label{cor:regimes}
Consider a DMC with $C>0$ and $C_1<\infty$. The following expansions hold as $N\to\infty$.
\begin{enumerate}
\item For every fixed $\epsilon\in(0,1)$,
\begin{align}
    \log \Mvlf(N,\epsilon) &= \frac{NC}{1-\epsilon}-\frac{C}{C_1}\log N+o(\log N). \label{eq:main-fixed-error}
\end{align}
\item For every fixed $\beta>0$,
\begin{align}
    \log \Mvlf(N,N^{-\beta}) &= \frac{NC}{1-N^{-\beta}}-\frac{C}{C_1}\max\{1,\beta\}\log N+o(\log N). \label{eq:main-polynomial}
\end{align}
\item For every fixed $E\in(0,C_1)$ and every sequence $\epsilon_N=\exp\{-EN+o(N)\}$,
\begin{align}
    \log \Mvlf(N,\epsilon_N) &= NC\left(1-\frac{E}{C_1}\right)+o(N). \label{eq:main-exponential}
\end{align}
\end{enumerate}
\end{corollary}

\begin{IEEEproof}
Each specialization satisfies the hypotheses of \thmref{thm:main} for all sufficiently large $N$; in part 3, choose $\eta\in(0,C_1-E)$. For fixed $\epsilon$, we have $A_N=\log N$ for all sufficiently large $N$. For $\epsilon_N=N^{-\beta}$, we have $A_N=\max\{1,\beta\}\log N$. Finally, if $\epsilon_N=\exp\{-EN+o(N)\}$ with $0<E<C_1$, then $A_N=EN+o(N)$ and $\frac{NC}{1-\epsilon_N}=NC+o(N)$. Substitution in~\eqref{eq:main-expansion} proves all three parts.
\end{IEEEproof}

These specializations relate the new converse to the existing asymptotic results as follows.
\begin{enumerate}
\item Part 1 gives the fundamental limit in the non-vanishing error probability regime.
\item Part 2 treats polynomially decaying error probabilities. The coefficient of $\log N$ in the fundamental limit is $-\frac{C}{C_1}$ for $\beta\leq1$ and $-\frac{\beta C}{C_1}$ for $\beta>1$.
\item Part 3 recovers Burnashev's error exponent for $0<E<C_1$. Writing the limiting rate as $R$,~\eqref{eq:main-exponential} gives $E=C_1\left(1-\frac{R}{C}\right)$, in agreement with~\eqref{eq:intro-burnashev}.
\end{enumerate}

The main theorem also has the following moderate-deviations specialization. Let $\{\rho_N\}_{N\in\mathbb N}$ be a positive rate-backoff sequence such that $\rho_N\to0$ and $\frac{N\rho_N}{\log N}\to\infty$. Suppose that
\begin{align}
    \log\frac{1}{\epsilon_N} &= \frac{C_1}{C}N\rho_N+o(N\rho_N). \label{eq:main-moderate-deviations-condition}
\end{align}
The stated conditions imply that the hypotheses of \thmref{thm:main} hold for all sufficiently large $N$, and~\eqref{eq:main-expansion} gives
\begin{align}
    \log \Mvlf(N,\epsilon_N) &= N(C-\rho_N)+o(N\rho_N). \label{eq:main-moderate-deviations}
\end{align}
The rate-backoff condition in~\cite[Th.~1]{TruongTan2019} implies $\frac{N\rho_N}{\log N}\to\infty$. Therefore,~\eqref{eq:main-moderate-deviations-condition} and~\eqref{eq:main-moderate-deviations} recover the optimal VLF moderate-deviations constant established there.

\subsection{Structure of Asymptotically Optimal Code Sequences}

The proof of the main converse also reveals necessary structure in code sequences with a non-vanishing error probability that achieve first-order or second-order optimality. In the following theorem, part (a) requires only positive capacity, whereas part (b) assumes $C_1<\infty$ and identifies the communication and confirmation behavior required for second-order optimality.

\begin{theorem}\label{thm:structure}
Fix $\epsilon\in(0,1)$ and consider any sequence of $(N_j,M_j,\epsilon)$-VLF codes over a DMC with $C>0$, where $N_j\to\infty$. For the $j$th code, use $U_j$ for the common randomness, $W_j$ for the message, $Y_j^t$ for the output history, $\tau_j$ for the stopping time, and $\ms{f}_{j,t}$ for the encoder at time $t$. Keep the encoder and stopping time of each code unchanged and replace only its terminal decoder by a deterministically tie-broken MAP decoder. On $\{\tau_j=t\}$,
\begin{align}
    \widehat W_j &= \min\left\{m\in[M_j]:\Prob{W_j=m\mid\mc{F}_{j,t}}=\max_{m'\in[M_j]}\Prob{W_j=m'\mid\mc{F}_{j,t}}\right\}, \qquad \mc{F}_{j,t}\triangleq\sigma(U_j,Y_j^t). \label{eq:main-sequence-map-decoder}
\end{align}
Let $P_{{\mr{e}},j}\triangleq\Prob{\widehat W_j\neq W_j}$. The MAP decoder minimizes the error probability for the fixed encoder and stopping time, so $P_{{\mr{e}},j}\leq\epsilon$. Define
\begin{align}
    \rho_{j,m}(t) &\triangleq \Prob{W_j=m\mid\mc{F}_{j,t}}, \qquad \rho_{j,m}(\tau_j)\triangleq\sum_{t=0}^{\infty}1\{\tau_j=t\}\rho_{j,m}(t), \label{eq:main-sequence-posterior}\\
    Z_j &\triangleq 1-\max_{m\in[M_j]}\rho_{j,m}(\tau_j), \qquad \mc{A}_j\triangleq\left\{Z_j>\frac{1}{2}\right\}, \qquad \mc{E}_j\triangleq\{\widehat W_j\neq W_j\}. \label{eq:main-sequence-events}
\end{align}
We write $\bs{\rho}_j(t)\triangleq(\rho_{j,1}(t),\ldots,\rho_{j,M_j}(t))$ and $\bs{\rho}_j(\tau_j)\triangleq(\rho_{j,1}(\tau_j),\ldots,\rho_{j,M_j}(\tau_j))$.
\par\smallskip\noindent\emph{(a) First-order terminal behavior and decoding times.}\par\noindent
Suppose that
\begin{align}
    \frac{\log M_j}{N_j} &\to \frac{C}{1-\epsilon}. \label{eq:main-first-order-codes}
\end{align}
Then,
\begin{align}
    P_{{\mr{e}},j} &\to\epsilon, \qquad \E{Z_j(1-Z_j)}\to0, \qquad \Prob{\mc{A}_j}\to\epsilon, \qquad \Prob{\mc{A}_j\mathbin{\triangle}\mc{E}_j}\to0, \label{eq:main-terminal-branches}\\
    \E{Z_j\mid\mc{A}_j^{\mr{c}}} &\to0, \qquad \E{1-Z_j\mid\mc{A}_j}\to0, \label{eq:main-branch-errors}\\
    \frac{\E{\tau_j}}{N_j} &\to1, \qquad \frac{\E{\tau_j1\{\mc{E}_j\}}}{N_j}\to0, \qquad \frac{\E{\tau_j1\{\mc{A}_j\}}}{N_j}\to0, \qquad \frac{\E{\tau_j1\{\mc{A}_j^{\mr{c}}\}}}{N_j}\to1. \label{eq:main-branch-times}
\end{align}
By~\eqref{eq:main-terminal-branches}, the events $\mc{E}_j$, $\mc{A}_j$, and $\mc{A}_j^{\mr{c}}$ have positive probability for all sufficiently large $j$. For every $\nu>0$,
\begin{align}
    \Prob{\tau_j>\nu N_j\mid\mc{E}_j} &\to0, \qquad \Prob{\tau_j>\nu N_j\mid\mc{A}_j}\to0. \label{eq:main-branch-tail-times}
\end{align}

\par\smallskip\noindent\emph{(b) Second-order communication and confirmation structure.}\par\noindent
Under the assumptions of part (a), suppose in addition that $C_1<\infty$ and
\begin{align}
    \log M_j &= \frac{CN_j}{1-\epsilon}-\frac{C}{C_1}\log N_j+o(\log N_j). \label{eq:main-second-order-codes}
\end{align}
For all sufficiently large $j$, define $r_j\triangleq(\log N_j\log\log N_j)^{-\frac{1}{3}}$. For $t\in\mathbb N$, define the following quantities, where $m\in[M_j]$ and $y\in\mc{Y}$.
\begin{align}
    P_{j,m,t}(y) &\triangleq P_{Y|X}\!\left(y\middle|\ms{f}_{j,t}(U_j,m,Y_j^{t-1})\right), \qquad \overline P_{j,t}\triangleq\sum_{m=1}^{M_j}\rho_{j,m}(t-1)P_{j,m,t}, \label{eq:main-sequence-output-laws}\\
    I_{j,t} &\triangleq \sum_{m=1}^{M_j}\rho_{j,m}(t-1)D(P_{j,m,t}\|\overline P_{j,t}), \label{eq:main-sequence-mutual-information}\\
    \operatorname{EJS}_{j,t} &\triangleq \operatorname{EJS}(\bs{\rho}_j(t-1);P_{j,1,t},\ldots,P_{j,M_j,t}), \label{eq:main-sequence-ejs}\\
    \mc{H}_{j,t}&\triangleq\left\{\max_{m\in[M_j]}\rho_{j,m}(t-1)>1-r_j\right\}, \label{eq:main-sequence-information-high}\\
    T_{j,\mr{conf}} &\triangleq \sum_{t=1}^{\infty}1\{\tau_j\geq t,\mc{H}_{j,t}\}. \label{eq:main-confirmation-occupancy}
\end{align}
For every $t\in\mathbb N$, the quantities in~\eqref{eq:main-sequence-output-laws}--\eqref{eq:main-sequence-information-high} are $\mc{F}_{j,t-1}$-measurable.
The assumption $C_1<\infty$ implies that all conditional output distributions have the same support. Hence, the finite-time message posteriors lie strictly between zero and one almost surely, and $\operatorname{EJS}_{j,t}$ is well-defined.
Then,
\begin{align}
    N_j-\E{\tau_j} &= o(\log N_j), \qquad \E{T_{j,\mr{conf}}}=\frac{1-\epsilon}{C_1}\log N_j+o(\log N_j), \label{eq:main-second-order-occupancies}\\
    \E{\sum_{t=1}^{\infty}1\{\tau_j\geq t,\mc{H}_{j,t}^{\mr{c}}\}(C-I_{j,t})} &= o(\log N_j), \label{eq:main-communication-efficiency}\\
    \E{\sum_{t=1}^{\infty}1\{\tau_j\geq t,\mc{H}_{j,t}\}I_{j,t}} &= o(\log N_j), \label{eq:main-confirmation-mutual-information}\\
    \E{\sum_{t=1}^{\infty}1\{\tau_j\geq t,\mc{H}_{j,t}\}\left(C_1-\operatorname{EJS}_{j,t}\right)}
    &= o(\log N_j). \label{eq:main-confirmation-efficiency}
\end{align}
\end{theorem}

\begin{IEEEproof}
See \secref{sec:proof-structure}.
\end{IEEEproof}

Part (a) applies to every first-order-optimal sequence, independently of its construction. The error constraint is asymptotically active, and the terminal conditional MAP error approaches either zero or one. The decoder-observable event $\mc{A}_j$ identifies the latter realizations and becomes equivalent in probability to the MAP-error event $\mc{E}_j$ (see~\eqref{eq:main-terminal-branches}). Its probability is $\epsilon+o(1)$, and, conditioned on this event, the decoder is almost always wrong and the mean decoding time is $o(N_j)$, although it need not stop at time zero. On $\mc{A}_j^{\mr{c}}$, the terminal decision is asymptotically reliable and the mean decoding time is $\frac{N_j}{1-\epsilon}+o(N_j)$.

Part (b) identifies the communication and confirmation structure required for second-order optimality. Over the low-posterior times $\mc{H}_{j,t}^{\mr{c}}$,~\eqref{eq:main-communication-efficiency} states that the expected stopped sum of the capacity deficits $C-I_{j,t}$ is $o(\log N_j)$. Over the high-posterior times $\mc{H}_{j,t}$,~\eqref{eq:main-second-order-occupancies}--\eqref{eq:main-confirmation-efficiency} show that the expected number of channel uses is $\frac{1-\epsilon}{C_1}\log N_j+o(\log N_j)$, while the expected stopped sums of $I_{j,t}$ and $C_1-\operatorname{EJS}_{j,t}$ are both $o(\log N_j)$. The set of high-posterior times may consist of several disjoint blocks.

\subsection{The BEC}

The finite-$C_1$ assumption in \thmref{thm:main} excludes channels whose conditional output distributions have different supports. For the BEC, we determine the exact minimum expected decoding time. Its input alphabet is $\{0,1\}$, its output alphabet is $\{0,1,\mathsf e\}$, and, for $x\in\{0,1\}$,
\begin{align}
    P_{Y|X}(x|x) &= 1-\delta, \qquad P_{Y|X}(\mathsf e|x)=\delta, \qquad \delta\in[0,1). \label{eq:main-bec-channel}
\end{align}
All other transition probabilities are zero.
For an integer $M\geq2$ and $\epsilon\in[0,1]$, let $\Nvlf(M,\epsilon;\delta)$ denote the minimum expected decoding time among VLF codes for $M$ equiprobable messages over this channel with average error probability at most $\epsilon$. Define
\begin{align}
    k_M &\triangleq \left\lfloor\log_2M\right\rfloor, \qquad \ell_{\mr{H}}(M)\triangleq k_M+2\left(1-\frac{2^{k_M}}{M}\right). \label{eq:main-huffman-length}
\end{align}
The quantity $\ell_{\mr{H}}(M)$ is the minimum average length of a binary prefix code for $M$ equiprobable messages~\cite{Huffman1952}.

\begin{theorem}\label{thm:bec}
For every integer $M\geq2$, $\delta\in[0,1)$, and $\epsilon\in[0,1-\frac{1}{M}]$,
\begin{align}
    \Nvlf(M,\epsilon;\delta) &= \frac{M(1-\epsilon)-1}{M-1}\frac{\ell_{\mr{H}}(M)}{1-\delta}. \label{eq:main-bec}
\end{align}
For $\epsilon\in[1-\frac{1}{M},1]$, we have $\Nvlf(M,\epsilon;\delta)=0$.
\end{theorem}

\begin{IEEEproof}
See \secref{sec:proof-bec}.
\end{IEEEproof}

For $\epsilon=0$, \thmref{thm:bec} agrees with~\cite[Cor.~6]{DevassyEtAl2016}. When $\epsilon=0$ and $M=2^k$ for an integer $k\geq1$, the value in~\eqref{eq:main-bec} equals the expected decoding time achieved by the construction in~\cite[Th.~7]{PolyanskiyPoorVerdu2011}. For $0<\epsilon\leq1-\frac{1}{M}$, the factor $\frac{M(1-\epsilon)-1}{M-1}$ in~\eqref{eq:main-bec} is the probability of proceeding with zero-error Huffman transmission when the code randomizes between that procedure and an immediate guess. In this range, the exact formula strictly strengthens the achievability bound in~\cite[Th.~3]{DevassyEtAl2016} and matches or strengthens the converse bound in~\cite[Th.~5]{DevassyEtAl2016}. For every fixed $\epsilon\in[0,1)$, \thmref{thm:bec} implies $\log \Mvlf(N,\epsilon)=\frac{NC}{1-\epsilon}+O(1)$ as $N\to\infty$.

\section{Supporting Results for the Converse} \label{sec:supporting}

Consider a DMC with $C>0$. We fix an arbitrary $(N,M,\epsilon)$-VLF code with $M\geq3$ and $\epsilon<1-\frac{1}{M}$. All quantities in this section refer to the causal encoders and the stopping time of this code. We first place the terminal observation on a fixed product space and derive the stopped Shannon-entropy identity. We then combine this identity with a stopped R\'enyi-entropy bound to show that a small Fano deficit forces the terminal conditional MAP error toward zero or one. These arguments apply without assuming $C_1<\infty$. Under $C_1<\infty$, we compare the Shannon-entropy and EJS drifts and thereby lower-bound the accumulated capacity deficit in terms of the terminal posterior potential.

For $m\in[M]$ and $t\in\mathbb N_0$, define the message posterior
\begin{align}
    \rho_m(t) &\triangleq \Prob{W=m\mid\mc{F}_t}, \label{eq:problem-posterior}
\end{align}
and write $\bs{\rho}(t)\triangleq(\rho_1(t),\ldots,\rho_M(t))$. We use the Bayes version on every positive-probability finite history and set the posterior to the uniform distribution on null histories. Since $\E{\tau}<\infty$, the stopping time is finite almost surely, and the random variables
\begin{align}
    \rho_m(\tau) &\triangleq \sum_{t=0}^{\infty}1\{\tau=t\}\rho_m(t), \qquad m\in[M], \label{eq:problem-stopped-posterior}
\end{align}
are well-defined. We write $\bs{\rho}(\tau)\triangleq(\rho_1(\tau),\ldots,\rho_M(\tau))$. To identify the information on which the terminal posterior is conditioned, choose a symbol $\mathsf T\notin\mc{Y}$ and define the padded outputs
\begin{align}
    \widetilde Y_t &\triangleq \begin{cases}Y_t,&\tau\geq t,\\ \mathsf T,&\tau<t,\end{cases}
    \qquad t\in\mathbb N, \label{eq:problem-padded-output}\\
    \mathsf{O}_\tau &\triangleq (U,\widetilde Y_1,\widetilde Y_2,\ldots)\in\mc{U}\times\left(\mc{Y}\cup\{\mathsf T\}\right)^{\mathbb N}, \qquad \mc{G}\triangleq\sigma(\mathsf{O}_\tau), \label{eq:problem-terminal-observation}
\end{align}
where the range of $\mathsf O_\tau$ carries its product sigma-field.
Each padded coordinate is measurable because $\{\tau\geq t\}\in\mc{F}_{t-1}$ and $Y_t$ is $\mc{F}_t$-measurable.

To prove the converse in~\cite[Th.~4]{PolyanskiyPoorVerdu2011}, Polyanskiy \emph{et al.} place an extended-channel output on a fixed infinite sequence. The padded observation $\mathsf O_\tau$ in~\eqref{eq:problem-terminal-observation} instead places decoder-side VLF stopping directly on a fixed product space, without changing the physical channel. The following lemma shows that $\mathsf O_\tau$ determines the stopping time and terminal estimate and that $\rho_m(\tau)$ is the posterior conditioned on this observation.

\begin{lemma}\label{lem:padded-observation}
The stopping time $\tau$ and every terminal estimate satisfying~\eqref{eq:problem-terminal-estimate} are $\mc{G}$-measurable. Moreover, for every $m\in[M]$,
\begin{align}
    \rho_m(\tau) &= \Prob{W=m\mid\mc{G}} \quad\text{almost surely}. \label{eq:problem-terminal-posterior}
\end{align}
\end{lemma}

\begin{IEEEproof}
The first padding symbol occurs at time $\tau+1$, and hence
\begin{align}
    \tau &= \inf\{t\in\mathbb N:\widetilde Y_t=\mathsf T\}-1=\sum_{t=1}^{\infty}1\{\widetilde Y_t\neq\mathsf T\} \quad\text{almost surely}. \label{eq:problem-recover-tau}
\end{align}
Thus, $\tau$ is $\mc{G}$-measurable. On $\{\tau=t\}$, the padded observation reveals $(U,Y^t)$ and hence the terminal value $\ms{g}_t(U,Y^t)$. Since each slice $\{\tau=t\}$ is $\mc{G}$-measurable, the terminal estimate is $\mc{G}$-measurable.

It remains to identify the posterior. For every $B\in\mc{G}$ and $t\in\mathbb N_0$,
\begin{align}
    B\cap\{\tau=t\} &\in\mc{F}_t. \label{eq:problem-terminal-slice}
\end{align}
Indeed, the claim holds first for every cylinder event generated by $U$ and finitely many padded coordinates. On $\{\tau=t\}$, the first $t$ padded output coordinates equal $Y^t$, and every later padded coordinate equals the deterministic symbol $\mathsf T$. The collection of $B$ satisfying~\eqref{eq:problem-terminal-slice} is a sigma-field, so the claim extends to $\mc{G}$.

Conversely, if $A\in\mc{F}_t$, then $A\cap\{\tau=t\}\in\mc{G}$: on $\{\tau=t\}$, the pair $(U,Y^t)$ equals $(U,\widetilde Y^t)$, and the event $\{\tau=t\}$ is recoverable from the padded coordinates. Since $\rho_m(t)$ is a measurable function of $(U,Y^t)$, each random variable $1\{\tau=t\}\rho_m(t)$ is $\mc{G}$-measurable. Their countable slice sum $\rho_m(\tau)$ is therefore $\mc{G}$-measurable.

We can now verify the conditional-expectation identity. For every $B\in\mc{G}$,
\begin{align}
    \E{1\{W=m\}1\{B\}} &= \sum_{t=0}^{\infty}\E{1\{W=m\}1\{B\cap\{\tau=t\}\}}=\sum_{t=0}^{\infty}\E{\rho_m(t)1\{B\cap\{\tau=t\}\}}=\E{\rho_m(\tau)1\{B\}}. \label{eq:problem-posterior-test}
\end{align}
This is the defining property of the conditional probability in~\eqref{eq:problem-terminal-posterior}.
\end{IEEEproof}

Define the terminal MAP estimate using deterministic tie-breaking by
\begin{align}
    \widehat W_{\mr{MAP}} &\triangleq \min\left\{ m\in[M]: \rho_m(\tau)=\max_{m'\in[M]}\rho_{m'}(\tau) \right\}, \label{eq:problem-map-estimate}
\end{align}
and define the conditional MAP error probability
\begin{align}
    \epsilon_{\mr{MAP}} &\triangleq 1-\max_{m\in[M]}\rho_m(\tau). \label{eq:problem-terminal-map-error}
\end{align}
Conditioned on $\mc{G}$, the prescribed decoder and the MAP decoder satisfy
\begin{align}
    \Prob{\widehat W\neq W\mid\mc{G}} &= 1-\rho_{\widehat W}(\tau) \geq \epsilon_{\mr{MAP}}, \label{eq:problem-original-conditional-error}\\
    \Prob{\widehat W_{\mr{MAP}}\neq W\mid\mc{G}} &= \epsilon_{\mr{MAP}}. \label{eq:problem-map-conditional-error}
\end{align}
Consequently,
\begin{align}
    \E{\epsilon_{\mr{MAP}}} &= \Prob{\widehat W_{\mr{MAP}}\neq W} \leq \Pe \leq \epsilon. \label{eq:problem-map-risk}
\end{align}

At time $t$, the encoder input under $W=m$ is the $\mc{F}_{t-1}$-measurable quantity
\begin{align}
    x_t(m) &\triangleq \ms{f}_t(U,m,Y^{t-1}), \qquad t\in\mathbb N. \label{eq:support-input-message}
\end{align}
For the realized history, define
\begin{align}
    P_{m,t}(y) &\triangleq P_{Y|X}\left(y\middle|x_t(m)\right), \qquad y\in\mc{Y}, \label{eq:support-message-output}\\
    \overline P_t(y) &\triangleq \sum_{m=1}^{M}\rho_m(t-1)P_{m,t}(y), \qquad y\in\mc{Y}. \label{eq:support-output-mixture}
\end{align}

The one-step conditional mutual information between the message and the next channel output is
\begin{align}
    I_t &\triangleq \sum_{m:\rho_m(t-1)>0}\rho_m(t-1)D\left(P_{m,t}\middle\|\overline P_t\right). \label{eq:support-mutual-information}
\end{align}
This is the $\mc{F}_{t-1}$-measurable one-step mutual information after fixing the past. Zero-posterior terms are omitted; under the finite-$C_1$ assumption used in the main converse, the restriction is vacuous.
For every realized past, $I_t$ is the mutual information induced by the corresponding conditional input distribution. Hence, $C-I_t\geq0$ almost surely.
For every $m$ such that $\rho_m(t-1)>0$, Bayes' rule gives
\begin{align}
    \rho_m(t) &= \frac{\rho_m(t-1)P_{m,t}(Y_t)}{\overline P_t(Y_t)}. 
\end{align}
Consequently, the posterior Shannon-entropy decrease is
\begin{align}
    \E{H(\bs{\rho}(t-1))-H(\bs{\rho}(t))\mid\mc{F}_{t-1}}
    &= \sum_{m:\rho_m(t-1)>0}\rho_m(t-1)
    \sum_{y:P_{m,t}(y)>0}P_{m,t}(y)
    \log\frac{\rho_m(t-1)P_{m,t}(y)}
    {\rho_m(t-1)\overline P_t(y)} \\
    &= \sum_{m:\rho_m(t-1)>0}\rho_m(t-1)D\left(P_{m,t}\middle\|\overline P_t\right)=I_t. \label{eq:support-shannon-drift}
\end{align}

\subsection{Terminal Posterior Concentration} \label{subsec:terminal-posterior}

Fix a capacity-achieving input distribution $P_X^*$ and let $P_Y^*$ denote the induced output distribution,
\begin{align}
    P_Y^*(y) &\triangleq \sum_{x\in\mc{X}}P_X^*(x)P_{Y|X}(y|x). \label{eq:notation-capacity-output}
\end{align}
By the capacity optimality condition in~\cite[Th.~4.5.1]{Gallager1968},
\begin{align}
    D\left(P_{Y|X=x}\middle\|P_Y^*\right) &\leq C, \qquad x\in\mc{X}. \label{eq:notation-capacity-condition}
\end{align}
Equality holds for every $x$ in the support of $P_X^*$.
The induced distribution $P_Y^*$ is the unique capacity-achieving output distribution~\cite[Cor.~2 to Th.~4.5.2]{Gallager1968}. Moreover, $P_Y^*(y)>0$ whenever $P_{Y|X}(y|x)>0$ for some $x\in\mc{X}$; otherwise, the divergence in~\eqref{eq:notation-capacity-condition} would be infinite.
For every $x\in\mc{X}$, define the information density on the support of $P_{Y|X=x}$ by
\begin{align}
    \imath(x;y) &\triangleq \log\frac{P_{Y|X}(y|x)}{P_Y^*(y)}, \qquad y\in\mc{Y}\text{ such that }P_{Y|X}(y|x)>0. \label{eq:notation-information-density}
\end{align}
Define the channel-dependent range constant by
\begin{align}
    b_{\mr{ch}} &\triangleq \max_{x\in\mc{X}} \left\{ \max_{y:P_{Y|X}(y|x)>0} \imath(x;y) - \min_{y:P_{Y|X}(y|x)>0} \imath(x;y) \right\} < \infty. \label{eq:notation-bch}
\end{align}

For $\alpha>1$, the closed-form expression for Sibson's order-$\alpha$ mutual information in~\cite[Cor.~2.3]{Sibson1969} is
\begin{align}
    I_\alpha^{\mr{S}}(P_X,P_{Y|X}) &\triangleq \frac{\alpha}{\alpha-1} \log \sum_{y\in\mc{Y}} \left( \sum_{x\in\mc{X}}P_X(x)P_{Y|X}(y|x)^\alpha \right)^{\frac{1}{\alpha}}. \label{eq:support-sibson-information}
\end{align}
We define the order-$\alpha$ Sibson capacity by
\begin{align}
    C_\alpha^{\mr{S}} &\triangleq \max_{P_X}I_\alpha^{\mr{S}}(P_X,P_{Y|X}). \label{eq:support-sibson-capacity}
\end{align}
For a recent overview of Sibson's $\alpha$-mutual information and its applications to statistical learning, hypothesis testing, and estimation theory, see~\cite{esposito2026Sibson}.
The proof uses two order-$\alpha$ bounds. The order-$\alpha$ Sibson capacity remains close to $C$ when $\alpha$ is close to one, while the terminal R\'enyi-entropy bound retains information about whether the conditional MAP error lies near an endpoint of $[0,1]$. The next two lemmas establish these properties.

\begin{lemma}\label{lem:renyi}
Consider an $(N,M,\epsilon)$-VLF code over a DMC, and let $\bs{\rho}(t)$ be its message-posterior process. For every $\alpha>1$,
\begin{align}
    \log M- \E{H_\alpha(\bs{\rho}(\tau))} &\leq C_\alpha^{\mr{S}}\E{\tau}. \label{eq:support-renyi-stopped}
\end{align}
Moreover, for every $s>0$,
\begin{align}
    C_{1+s}^{\mr{S}} &\leq C+\frac{s b_{\mr{ch}}^2}{8}. \label{eq:support-sibson-capacity-bound}
\end{align}
\end{lemma}

\begin{IEEEproof}
See \appref{app:renyi}.
\end{IEEEproof}

For $M\geq3$, define
\begin{align}
    \alpha_M &\triangleq 1+\frac{2}{\log(M-1)}, \label{eq:support-alpha-M}\\
    \kappa &\triangleq 1+\log\frac{1+\exp\{-2\}}{2}>0, \label{eq:support-kappa}\\
    c_0 &\triangleq \exp\{1\}+\exp\{-1\}. \label{eq:support-c-zero}
\end{align}
We next bound the terminal R\'enyi entropy pointwise in terms of the MAP error.

\begin{lemma}\label{lem:terminal-renyi}
Let $M\geq3$, let $\alpha_M$, $\kappa$, and $c_0$ be given by~\eqref{eq:support-alpha-M}--\eqref{eq:support-c-zero}, and let $\bs{\rho}$ be a probability vector on $[M]$. Set $\theta\triangleq1-\max_{m\in[M]}\rho_m$. Then,
\begin{align}
    H_{\alpha_M}(\bs{\rho}) &\leq \theta\log(M-1)-\kappa\log(M-1)\theta(1-\theta)+c_0. \label{eq:support-terminal-renyi}
\end{align}
\end{lemma}

\begin{IEEEproof}
See \appref{app:terminal-renyi}.
\end{IEEEproof}

The magnitude of the negative quadratic term in~\eqref{eq:support-terminal-renyi} is of order $\log(M-1)$ whenever the terminal conditional MAP error remains bounded away from zero and one.

For code parameters $(N,M,\epsilon)$, define the Fano deficit
\begin{align}
    \Delta_{\mr{F}}(N,M,\epsilon) &\triangleq CN+h_{\mr{b}}(\epsilon)-(1-\epsilon)\log M. \label{eq:support-delta}
\end{align}
The upper bound in~\eqref{eq:intro-ppv} shows that $\Delta_{\mr{F}}(N,M,\epsilon)$ is non-negative. Our goal is to derive a sharper lower bound on this quantity.
The following lemma first establishes the exact stopped Shannon-entropy identity and then combines it with \lemref{lem:renyi} and \lemref{lem:terminal-renyi}. It shows that a small Fano deficit forces the terminal conditional MAP error toward the endpoints of $[0,1]$.

\begin{lemma}\label{lem:terminal-concentration}
Consider an $(N,M,\epsilon)$-VLF code over a DMC with $C>0$, $M\geq3$, and $\epsilon<1-\frac{1}{M}$. Then, the terminal conditional MAP error satisfies
\begin{align}
    \E{\epsilon_{\mr{MAP}}(1-\epsilon_{\mr{MAP}})} &\leq \frac{ \frac{b_{\mr{ch}}^2N}{4\log(M-1)} +\Delta_{\mr{F}}(N,M,\epsilon)+c_0 }{ \kappa\log(M-1) }. \label{eq:support-terminal-concentration}
\end{align}
\end{lemma}

\begin{IEEEproof}
We first record the exact Shannon-entropy identity for the fixed observation in~\eqref{eq:problem-terminal-observation}. The indicator $1\{\tau\geq t\}$ is $\sigma(U,\widetilde Y^{t-1})$-measurable. If a padding symbol occurs in $\widetilde Y^{t-1}$, then $\tau<t$. Otherwise, $(U,\widetilde Y^{t-1})=(U,Y^{t-1})$, and the stopping rule determines whether $\tau=t-1$ or $\tau\geq t$. On $\{\tau\geq t\}$, the conditional law of $(W,\widetilde Y_t)$ given the padded past is the conditional law of $(W,Y_t)$ given $\mc{F}_{t-1}$; on the complementary event, $\widetilde Y_t=\mathsf T$ is deterministic. Hence,
\begin{align}
    I(W;\widetilde Y_t\mid U,\widetilde Y^{t-1}) &= \E{1\{\tau\geq t\}I_t}. \label{eq:support-padded-one-step}
\end{align}
For each finite $n$, the ordinary conditional chain rule sums~\eqref{eq:support-padded-one-step} from $t=1$ to $n$. Since $\sigma(U,\widetilde Y^n)$ increases to $\mc{G}$, L\'evy's upward theorem~\cite[Sec.~14.2]{Williams1991} gives, for every $m\in[M]$,
\begin{align}
    \Prob{W=m\mid U,\widetilde Y^n} &\longrightarrow \Prob{W=m\mid\mc{G}}\quad\text{almost surely}. \label{eq:support-padded-posterior-limit}
\end{align}
Continuity and boundedness of entropy then justify passage to the limit. Since $W$ is independent of $U$ and $\mathsf O_\tau$ contains $U$, we obtain
\begin{align}
    I(W;\mathsf{O}_\tau) &= \lim_{n\to\infty}I(W;\widetilde Y^n\mid U) = \sum_{t=1}^{\infty}I(W;\widetilde Y_t\mid U,\widetilde Y^{t-1}) = \E{\sum_{t=1}^{\infty}1\{\tau\geq t\}I_t}. \label{eq:support-padded-information}
\end{align}
Equation~\eqref{eq:support-padded-information} is the stopped mutual-information chain rule for decoder-side VLF stopping. By \lemref{lem:padded-observation},
\begin{align}
    \log M-\E{H(\bs{\rho}(\tau))} &= I(W;\mathsf{O}_\tau)=\E{\sum_{t=1}^{\infty}1\{\tau\geq t\}I_t}. \label{eq:support-stopped-shannon}
\end{align}
Since $\widehat W$ is $\mc{G}$-measurable, Fano's inequality,~\eqref{eq:support-stopped-shannon}, and $\log(M-1)\leq\log M$ give
\begin{align}
    H(W\mid\mathsf{O}_\tau) &\leq H(W\mid\widehat W)\leq h_{\mr{b}}(\Pe)+\Pe\log(M-1), \label{eq:support-fano-entropy}\\
    (1-\Pe)\log M &\leq I(W;\mathsf{O}_\tau)+h_{\mr{b}}(\Pe)\leq C\E{\tau}+h_{\mr{b}}(\Pe)\leq CN+h_{\mr{b}}(\Pe). \label{eq:support-fano}
\end{align}
Using~\eqref{eq:support-stopped-shannon} and $\E{\tau}=\E{\sum_{t=1}^{\infty}1\{\tau\geq t\}}$, the Fano deficit has the exact decomposition
\begin{align}
    \Delta_{\mr{F}}(N,M,\epsilon) &= C\left(N-\E{\tau}\right)+\E{\sum_{t=1}^{\infty}1\{\tau\geq t\}(C-I_t)} +\left[h_{\mr{b}}(\epsilon)+\epsilon\log M-\E{H(\bs{\rho}(\tau))}\right]. \label{eq:support-delta-decomposition}
\end{align}
The first two terms on the right-hand side of the equality in~\eqref{eq:support-delta-decomposition} are non-negative because $\E{\tau}\leq N$ and $I_t\leq C$. Taking $\widehat W=\widehat W_{\mr{MAP}}$ in~\eqref{eq:support-fano-entropy}, and using~\eqref{eq:problem-map-risk} together with the monotonicity of $a\mapsto h_{\mr{b}}(a)+a\log(M-1)$ on $[0,1-\frac{1}{M}]$, give
\begin{align}
    \E{H(\bs{\rho}(\tau))}=H(W\mid\mathsf{O}_\tau) &\leq h_{\mr{b}}(\epsilon)+\epsilon\log(M-1)\leq h_{\mr{b}}(\epsilon)+\epsilon\log M. 
\end{align}
Thus, the bracketed third term in~\eqref{eq:support-delta-decomposition} is also non-negative.

We apply \lemref{lem:renyi} with the order in~\eqref{eq:support-alpha-M}. Subtracting~\eqref{eq:support-stopped-shannon} from~\eqref{eq:support-renyi-stopped}, and then using~\eqref{eq:support-sibson-capacity-bound} and $\E{\tau}\leq N$, gives
\begin{align}
    \E{H(\bs{\rho}(\tau))-H_{\alpha_M}(\bs{\rho}(\tau))} &\leq \frac{b_{\mr{ch}}^2N}{4\log(M-1)}+\E{\sum_{t=1}^{\infty}1\{\tau\geq t\}(C-I_t)}. \label{eq:support-shannon-renyi-difference}
\end{align}
Solving~\eqref{eq:support-delta-decomposition} for $\E{H(\bs{\rho}(\tau))}$, subtracting the bound in~\eqref{eq:support-shannon-renyi-difference}, and discarding the non-negative quantity $C(N-\E{\tau})$ give
\begin{align}
    \E{H_{\alpha_M}(\bs{\rho}(\tau))} &\geq h_{\mr{b}}(\epsilon)+\epsilon\log M-\frac{b_{\mr{ch}}^2N}{4\log(M-1)}-\Delta_{\mr{F}}(N,M,\epsilon). \label{eq:support-terminal-renyi-lower}
\end{align}
On the other hand, \lemref{lem:terminal-renyi} and~\eqref{eq:problem-map-risk} yield
\begin{align}
    \E{H_{\alpha_M}(\bs{\rho}(\tau))} &\leq \epsilon\log(M-1)-\kappa\log(M-1)\E{\epsilon_{\mr{MAP}}(1-\epsilon_{\mr{MAP}})}+c_0. \label{eq:support-terminal-renyi-upper}
\end{align}
We combine~\eqref{eq:support-terminal-renyi-lower} and~\eqref{eq:support-terminal-renyi-upper}. After rearrangement, the upper bound on $\kappa\log(M-1)\E{\epsilon_{\mr{MAP}}(1-\epsilon_{\mr{MAP}})}$ contains the two non-positive contributions $-h_{\mr{b}}(\epsilon)$ and $-\epsilon\log\frac{M}{M-1}$. Discarding them proves~\eqref{eq:support-terminal-concentration}.
\end{IEEEproof}

The preceding endpoint-concentration estimate bounds $\E{\epsilon_{\mr{MAP}}(1-\epsilon_{\mr{MAP}})}$, whereas the lower bound on the terminal posterior potential in~\eqref{eq:support-L-lower} involves $\E{(1-\epsilon_{\mr{MAP}})\log\frac{1}{\epsilon_{\mr{MAP}}}}$. The following scalar inequality connects these quantities.

\begin{lemma}\label{lem:scalar}
Let $Z$ take values in $[0,1]$, suppose that $\E{Z}\leq\epsilon$, and define
\begin{align}
    v_Z &\triangleq \E{Z(1-Z)}. \label{eq:support-v-Z}
\end{align}
If $0<v_Z<\frac{1-\epsilon}{4}$, then
\begin{align}
    \E{(1-Z)\log\frac{1}{Z}} &\geq (1-\epsilon-2v_Z) \log\frac{1-\epsilon-2v_Z}{2v_Z} -\exp\{-1\}. \label{eq:support-scalar}
\end{align}
\end{lemma}

\begin{IEEEproof}
See \appref{app:scalar}.
\end{IEEEproof}

\subsection{Comparison of the Shannon-Entropy and EJS Drifts} \label{subsec:posterior-log-odds}

For the remainder of this section, suppose that $C_1<\infty$. All conditional output distributions then have the same support. Indeed, if an output symbol has positive probability under one input and zero probability under another, one of the two ordered KL divergences in~\eqref{eq:notation-C1} is infinite. After deleting output symbols that have zero probability under every input, we have $P_{Y|X}(y|x)>0$ for all $(x,y)\in\mc{X}\times\mc{Y}$. We define
\begin{align}
    \lambda_{\min} &\triangleq \min_{x,x'\in\mc{X}} \min_{y\in\mc{Y}} \frac{P_{Y|X}(y|x)}{P_{Y|X}(y|x')} >0. \label{eq:notation-lambda-min}
\end{align}
Conditional on every positive-probability value of $U$, common support makes every attainable finite output history have positive likelihood under every message. The uniform prior and Bayes' rule therefore imply that $0<\rho_m(t)<1$ for every $m\in[M]$ and every finite $t$, almost surely. On null histories this property also holds for the uniform posterior assigned after~\eqref{eq:problem-posterior}. We may therefore define the output distribution conditioned on the event that message $m$ is incorrect as
\begin{align}
    Q_{m,t}(y) &\triangleq \frac{\overline P_t(y)-\rho_m(t-1)P_{m,t}(y)}{1-\rho_m(t-1)}. \label{eq:support-extrinsic-output}
\end{align}
By~\eqref{eq:notation-ejs}, the EJS divergence for the posterior and one-step output laws at time $t$ is
\begin{align}
    \operatorname{EJS}(\bs{\rho}(t-1);P_{1,t},\ldots,P_{M,t}) &= \sum_{m=1}^{M}\rho_m(t-1)D\left(P_{m,t}\middle\|Q_{m,t}\right). \label{eq:support-ejs}
\end{align}
In~\cite[Lem.~2]{NaghshvarJavidiWigger2015}, Naghshvar \emph{et al.} prove that the EJS divergence is at least the Jensen--Shannon divergence, which equals $I_t$ by~\eqref{eq:support-mutual-information}. Therefore,
\begin{align}
    I_t &\leq \operatorname{EJS}(\bs{\rho}(t-1);P_{1,t},\ldots,P_{M,t}). \label{eq:support-ejs-dominates-information}
\end{align}

We define the posterior potential
\begin{align}
    L(\bs{\rho}) &\triangleq \sum_{m=1}^{M}\rho_m\log\frac{1}{1-\rho_m}. \label{eq:support-L}
\end{align}
Equivalently,
\begin{align}
    L(\bs{\rho}) &= H(\bs{\rho})+\sum_{m=1}^{M}\rho_m\log\frac{\rho_m}{1-\rho_m}. \label{eq:support-L-decomposition}
\end{align}
Every summand in~\eqref{eq:support-L} is non-negative. Keeping only a summand corresponding to a MAP message gives
\begin{align}
    L(\bs{\rho}(\tau)) &\geq (1-\epsilon_{\mr{MAP}})\log\frac{1}{\epsilon_{\mr{MAP}}}. \label{eq:support-L-lower}
\end{align}
The main converse now relates this terminal posterior potential to the accumulated capacity deficit. The next lemma performs the comparison for every causal encoder action by expressing the drift of $L$ as the EJS divergence minus the mutual information and bounding this difference in terms of $C-I_t$.

\begin{lemma}\label{lem:mutual-information-difference}
Consider an $(N,M,\epsilon)$-VLF code with $M\geq3$ over a DMC with $C>0$ and $C_1<\infty$, and let $I_t$ and $L$ be defined by~\eqref{eq:support-mutual-information} and~\eqref{eq:support-L}, respectively. For every $r\in(0,\frac{1}{2})$ such that
\begin{align}
    h_{\mr{b}}(r)+r\log(|\mc{X}|-1) &< C, \label{eq:support-r-condition}
\end{align}
we have
\begin{align}
    \E{\sum_{t=1}^{\infty}1\{\tau\geq t\}(C-I_t)} &\geq \frac{C-h_{\mr{b}}(r)-r\log(|\mc{X}|-1)}{C_1}\left(\E{L(\bs{\rho}(\tau))}-\log\frac{M}{M-1}-\frac{1}{r}-\frac{1}{2\lambda_{\min}^{3}r^2}\right)^+. \label{eq:support-mutual-information-difference}
\end{align}
\end{lemma}

\begin{IEEEproof}
The posterior average log-odds drift identity in~\cite[Lem.~2]{NaghshvarJavidiWigger2015}, together with~\eqref{eq:support-L-decomposition} and the Shannon-entropy drift in~\eqref{eq:support-shannon-drift}, gives
\begin{align}
    \E{L(\bs{\rho}(t))-L(\bs{\rho}(t-1))\mid\mc{F}_{t-1}} &= \operatorname{EJS}(\bs{\rho}(t-1);P_{1,t},\ldots,P_{M,t})-I_t. \label{eq:support-L-drift}
\end{align}
We bound the right-hand side of the identity in~\eqref{eq:support-L-drift} in the two posterior regimes determined by the threshold $1-r$, and then stop and sum the resulting bounds. Let
\begin{align}
    V(\bs{\rho}) &\triangleq \sum_{m=1}^{M}\rho_m^2. \label{eq:support-posterior-square}
\end{align}
Bayes' rule and the definition of the $\chi^2$-divergence give
\begin{align}
    \E{V(\bs{\rho}(t))-V(\bs{\rho}(t-1))\mid\mc{F}_{t-1}}
    &= \sum_{m=1}^{M}\rho_m(t-1)^2\left(\sum_{y\in\mc{Y}}\frac{P_{m,t}(y)^2}{\overline P_t(y)}-1\right) = \sum_{m=1}^{M}\rho_m(t-1)^2 \chi^2\left(P_{m,t}\middle\|\overline P_t\right). \label{eq:support-posterior-square-drift}
\end{align}

First, suppose that
\begin{align}
    \max_{m\in[M]}\rho_m(t-1) &\leq 1-r. \label{eq:support-no-large-posterior}
\end{align}
In this region, define
\begin{align}
    \ell_{m,t}(y) &\triangleq \frac{P_{m,t}(y)}{\overline P_t(y)}, \qquad u_{m,t}(y)\triangleq \frac{\rho_m(t-1)\left(\ell_{m,t}(y)-1\right)}{1-\rho_m(t-1)}. \label{eq:support-local-ratios}
\end{align}
Equation~\eqref{eq:support-extrinsic-output} gives $\frac{Q_{m,t}(y)}{\overline P_t(y)}=1-u_{m,t}(y)$.
The likelihood-ratio bound in~\eqref{eq:notation-lambda-min} is inherited by mixtures. Averaging the bound over the input mixtures that induce $Q_{m,t}$ and $\overline P_t$ gives $Q_{m,t}(y)\geq\lambda_{\min}\overline P_t(y)$. Averaging it over the mixture inducing $\overline P_t$, with $P_{m,t}$ fixed on the right, gives $\overline P_t(y)\geq\lambda_{\min}P_{m,t}(y)$. Hence,
\begin{align}
    1-u_{m,t}(y) &\geq \lambda_{\min}, \qquad \ell_{m,t}(y)\leq \frac{1}{\lambda_{\min}}. \label{eq:support-ratio-bounds}
\end{align}
Since $\lambda_{\min}\leq1$, the first inequality in~\eqref{eq:support-ratio-bounds} gives $1-v\geq\lambda_{\min}$ for every $v$ between $0$ and $u_{m,t}(y)$, so Taylor's theorem applies for either sign of $u_{m,t}(y)$ and gives
\begin{align}
    -\log(1-u_{m,t}(y)) &\leq u_{m,t}(y)+\frac{u_{m,t}(y)^2}{2\lambda_{\min}^2}, \qquad y\in\mc{Y}. \label{eq:support-log-bound}
\end{align}
For $Y\sim P_{m,t}$, we have
\begin{align}
    \E{\ell_{m,t}(Y)-1} &= \chi^2\left(P_{m,t}\middle\|\overline P_t\right), \\
    \E{(\ell_{m,t}(Y)-1)^2} &\leq \frac{1}{\lambda_{\min}}\chi^2\left(P_{m,t}\middle\|\overline P_t\right), \label{eq:support-moments}
\end{align}
where the inequality uses $\ell_{m,t}(y)\leq\frac{1}{\lambda_{\min}}$.
Moreover,
\begin{align}
    D\left(P_{m,t}\middle\|Q_{m,t}\right)-D\left(P_{m,t}\middle\|\overline P_t\right) &= \E{-\log(1-u_{m,t}(Y))}. \label{eq:support-divergence-difference}
\end{align}
Substituting the definition of $u_{m,t}$ into~\eqref{eq:support-log-bound}, taking the $P_{m,t}$-expectation, and applying~\eqref{eq:support-moments} in~\eqref{eq:support-divergence-difference} gives
\begin{align}
    D\left(P_{m,t}\middle\|Q_{m,t}\right)-D\left(P_{m,t}\middle\|\overline P_t\right) &\leq \Biggl[\frac{\rho_m(t-1)}{1-\rho_m(t-1)} +\frac{\rho_m(t-1)^2}{2\lambda_{\min}^3(1-\rho_m(t-1))^2}\Biggr]\chi^2\left(P_{m,t}\middle\|\overline P_t\right). \label{eq:support-one-message-bound}
\end{align}
Multiplying~\eqref{eq:support-one-message-bound} by $\rho_m(t-1)$, summing over $m$, and using~\eqref{eq:support-no-large-posterior} and~\eqref{eq:support-posterior-square-drift}, we obtain
\begin{align}
    \operatorname{EJS}(\bs{\rho}(t-1);P_{1,t},\ldots,P_{M,t})-I_t &\leq \left( \frac{1}{r}+\frac{1}{2\lambda_{\min}^3r^2} \right) \E{V(\bs{\rho}(t))-V(\bs{\rho}(t-1))\mid\mc{F}_{t-1}}. \label{eq:support-no-large-posterior-bound}
\end{align}

Next, suppose that $\max_{m\in[M]}\rho_m(t-1)>1-r$. Let $P_{X,t}$ denote the induced input distribution
\begin{align}
    P_{X,t}(x) &\triangleq \sum_{m=1}^{M}\rho_m(t-1)\,1\{x_t(m)=x\}, \qquad x\in\mc{X}. \label{eq:support-input-distribution}
\end{align}
Set
\begin{align}
    \vartheta_r &\triangleq h_{\mr{b}}(r)+r\log(|\mc{X}|-1)<C. \label{eq:support-vartheta-r}
\end{align}
If $\rho_{m_*}(t-1)>1-r$, then $P_{X,t}(x_t(m_*))\geq\rho_{m_*}(t-1)>1-r$. Hence, the mass outside a most likely input symbol is smaller than $r$. For a fixed outside mass $a$, distributing it uniformly over the other $|\mc{X}|-1$ symbols maximizes the entropy at $h_{\mr{b}}(a)+a\log(|\mc{X}|-1)$. This expression is increasing for $a\in[0,r]$ because $r<\frac{1}{2}$, and therefore
\begin{align}
    I_t &\leq H(P_{X,t})\leq\vartheta_r. \label{eq:support-large-posterior-mutual-information}
\end{align}
For every $m$, $Q_{m,t}$ is the posterior-weighted convex combination of $P_{m',t}$ over $m'\neq m$. Convexity of KL divergence in its second argument and the definition of $C_1$ give $D(P_{m,t}\|Q_{m,t})\leq C_1$. Averaging over $m$ yields
\begin{align}
    \operatorname{EJS}(\bs{\rho}(t-1);P_{1,t},\ldots,P_{M,t}) &\leq C_1. \label{eq:support-ejs-C-one}
\end{align}
Since $I_t\leq\vartheta_r<C$, the comparison between $C_1-I_t$ and $C-I_t$ follows from
\begin{align}
    C_1(C-I_t)-(C_1-I_t)(C-\vartheta_r) &= C_1(\vartheta_r-I_t)+I_t(C-\vartheta_r)\geq0. \label{eq:support-large-posterior-algebra}
\end{align}
Thus, every posterior satisfying $\max_{m\in[M]}\rho_m(t-1)>1-r$ also satisfies
\begin{align}
    \operatorname{EJS}(\bs{\rho}(t-1);P_{1,t},\ldots,P_{M,t})-I_t &\leq C_1-I_t \leq \frac{C_1}{C-\vartheta_r}(C-I_t). \label{eq:support-large-posterior-bound}
\end{align}

Finally, truncate the stopping time and sum the two bounds. For $k\in\mathbb N$, let $\tau^{(k)}\triangleq\min\{\tau,k\}$. The indicators $1\{\tau^{(k)}\geq t\}$, $1\{\max_{m\in[M]}\rho_m(t-1)\leq1-r\}$, and $1\{\max_{m\in[M]}\rho_m(t-1)>1-r\}$ are $\mc{F}_{t-1}$-measurable, so they may multiply the corresponding conditional drift estimates before expectation. The common-support property and the finiteness of $\mc{U}$ and $\mc{Y}$ imply that $L(\bs{\rho}(t))$ is bounded over the histories at every fixed $t$. Stopping~\eqref{eq:support-L-drift} at $\tau^{(k)}$ gives
\begin{align}
    \E{L(\bs{\rho}(\tau^{(k)}))}-L(\bs{\rho}(0)) &= \E{\sum_{t=1}^{k}1\{\tau^{(k)}\geq t\}\left[\operatorname{EJS}(\bs{\rho}(t-1);P_{1,t},\ldots,P_{M,t})-I_t\right]}. \label{eq:support-stopped-L}
\end{align}
For the terms satisfying~\eqref{eq:support-no-large-posterior}, we apply~\eqref{eq:support-no-large-posterior-bound}. Their contribution to the non-negative drift in~\eqref{eq:support-posterior-square-drift} is at most its full stopped sum, which telescopes according to
\begin{align}
    \E{\sum_{t=1}^{k}1\{\tau^{(k)}\geq t\}\E{V(\bs{\rho}(t))-V(\bs{\rho}(t-1))\mid\mc{F}_{t-1}}} &= \E{V(\bs{\rho}(\tau^{(k)}))}-V(\bs{\rho}(0))\leq1. \label{eq:support-stopped-posterior-square}
\end{align}
For the remaining terms, we apply~\eqref{eq:support-large-posterior-bound}. Equations~\eqref{eq:support-stopped-L} and~\eqref{eq:support-stopped-posterior-square} then give
\begin{align}
    \E{L(\bs{\rho}(\tau^{(k)}))}-L(\bs{\rho}(0)) &\leq \frac{1}{r}+\frac{1}{2\lambda_{\min}^3r^2}+\frac{C_1}{C-\vartheta_r}\E{\sum_{t=1}^{k}1\{\tau^{(k)}\geq t\}(C-I_t)}. \label{eq:support-stopped-comparison}
\end{align}
Since $\tau<\infty$ almost surely, $L(\bs{\rho}(\tau^{(k)}))$ converges almost surely to $L(\bs{\rho}(\tau))$. Fatou's lemma gives the first inequality below. For every $k$,~\eqref{eq:support-stopped-comparison} bounds the expression inside the $\liminf$, and monotone convergence gives the second inequality
\begin{align}
    \E{L(\bs{\rho}(\tau))}-L(\bs{\rho}(0))
    &\leq \liminf_{k\to\infty}\left[\E{L(\bs{\rho}(\tau^{(k)}))}-L(\bs{\rho}(0))\right] \\
    &\leq \frac{1}{r}+\frac{1}{2\lambda_{\min}^3r^2}+\frac{C_1}{C-\vartheta_r}\E{\sum_{t=1}^{\infty}1\{\tau\geq t\}(C-I_t)}.
\end{align}
Since the initial posterior is uniform,
\begin{align}
    L(\bs{\rho}(0)) &= \log\frac{M}{M-1}. \label{eq:support-initial-L}
\end{align}
Rearrangement gives the lower bound in~\eqref{eq:support-mutual-information-difference} without the positive-part operator. Because the stopped capacity deficit on the left-hand side is non-negative, taking the maximum of this lower bound and zero yields the positive-part form in~\eqref{eq:support-mutual-information-difference}.
\end{IEEEproof}

\section{Proof of \thmref{thm:main}} \label{sec:proof-main}

For the converse, suppose that the uniform bound in~\eqref{eq:main-converse} fails. Then, there exist a constant $d>0$, a strictly increasing sequence of positive integers $N_j\to\infty$, error probabilities $\epsilon_j$ satisfying
\begin{align}
    0<\epsilon_j &\leq \bar\epsilon, \qquad \log\frac{1}{\epsilon_j}\leq (C_1-\eta)N_j, \label{eq:proof-error-sequence-conditions}
\end{align}
and an $(N_j,M_j,\epsilon_j)$-VLF code for every $j\in\mathbb N$ such that, with $A_j\triangleq\max\{\log N_j,\log\frac{1}{\epsilon_j}\}$,
\begin{align}
    \log M_j &\geq \frac{N_jC+h_{\mr{b}}(\epsilon_j)}{1-\epsilon_j}-\frac{C}{C_1}A_j+dA_j. \label{eq:proof-contradiction-assumption}
\end{align}
For the $j$th code, let $\tau_j$ be its stopping time and let $\bs{\rho}_j(t)$ be its posterior process. Write $I_{j,t}$ for the conditional mutual information in~\eqref{eq:support-mutual-information} for this code, and define $\epsilon_{\mr{MAP},j}\triangleq1-\max_{m\in[M_j]}\rho_{j,m}(\tau_j)$.

The proof first determines the scale of the message set and derives an upper bound on the Fano deficit from the contradiction assumption. It then uses the terminal posterior potential to derive a conflicting lower bound on the stopped capacity deficit.

Applying the Fano-capacity bound in~\eqref{eq:support-fano} to the terminal MAP estimate and using $\E{\epsilon_{\mr{MAP},j}}\leq\epsilon_j$ and $h_{\mr{b}}(a)\leq\log 2$ give
\begin{align}
    (1-\epsilon_j)\log M_j &\leq CN_j+\log 2. \label{eq:proof-Fano-order}
\end{align}
For all sufficiently large $j$, $\log N_j\leq(C_1-\eta)N_j$. Together with~\eqref{eq:proof-error-sequence-conditions}, this gives $A_j\leq(C_1-\eta)N_j$. Assumption~\eqref{eq:proof-contradiction-assumption} then implies $\log M_j=\Omega(N_j)$, while~\eqref{eq:proof-Fano-order} implies $\log M_j=O(N_j)$. Hence,
\begin{align}
    \log(M_j-1) &= \Theta(N_j). \label{eq:proof-log-M-order}
\end{align}
In particular, $M_j\geq3$ and $\epsilon_j<1-\frac{1}{M_j}$ for all sufficiently large $j$, so the finite-length estimates in \secref{sec:supporting} apply.

We next introduce the Fano deficit for the $j$th code.
\begin{align}
    \Delta_j &\triangleq \Delta_{\mr{F}}(N_j,M_j,\epsilon_j). \label{eq:proof-delta-j}
\end{align}
The exact decomposition in~\eqref{eq:support-delta-decomposition} yields
\begin{align}
    \Delta_j &\geq \E{\sum_{t=1}^{\infty}1\{\tau_j\geq t\}(C-I_{j,t})}. \label{eq:proof-delta-lower}
\end{align}
The definition of $\Delta_j$ also gives
\begin{align}
    (1-\epsilon_j)\log M_j &= CN_j+h_{\mr{b}}(\epsilon_j)-\Delta_j. \label{eq:proof-delta-identity}
\end{align}
Combining~\eqref{eq:proof-delta-identity} with~\eqref{eq:proof-contradiction-assumption} and using $1-\epsilon_j\geq1-\bar\epsilon$ give
\begin{align}
    \Delta_j &\leq \frac{C}{C_1}(1-\epsilon_j)A_j-d(1-\bar\epsilon)A_j. \label{eq:proof-delta-upper}
\end{align}
Thus, $\Delta_j=O(A_j)$; in particular,~\eqref{eq:proof-delta-lower} shows that the stopped capacity deficit is $O(A_j)$.

We next derive a lower bound of order $A_j$ on the expected terminal posterior potential. Let
\begin{align}
    \mc{J}_{\mr{L}} &\triangleq \left\{j:\log\frac{1}{\epsilon_j}\leq\log N_j\right\}, \qquad \mc{J}_{\mr{E}}\triangleq \left\{j:\log\frac{1}{\epsilon_j}>\log N_j\right\}. \label{eq:proof-index-partition}
\end{align}
For $j\in\mc{J}_{\mr{L}}$, we have $A_j=\log N_j$ and~\eqref{eq:proof-delta-upper} gives $\Delta_j=O(\log N_j)$. By \lemref{lem:terminal-concentration} and~\eqref{eq:proof-log-M-order},
\begin{align}
    v_j \triangleq \E{\epsilon_{\mr{MAP},j}(1-\epsilon_{\mr{MAP},j})} &= O\left(\frac{\log N_j}{N_j}\right). \label{eq:proof-v-j}
\end{align}
The common-support property implies $v_j>0$, while~\eqref{eq:proof-v-j} and $1-\epsilon_j\geq1-\bar\epsilon>0$ give $v_j<\frac{1-\epsilon_j}{4}$ for all sufficiently large $j$. Equation~\eqref{eq:proof-v-j} also gives $\log\frac{1}{v_j}\geq\log N_j-\log\log N_j-O(1)$ and $v_j\log\frac{1}{v_j}=o(1)$. Together with $\E{\epsilon_{\mr{MAP},j}}\leq\epsilon_j$, these estimates verify the hypotheses of \lemref{lem:scalar}. Applying that lemma with $Z=\epsilon_{\mr{MAP},j}$ and then using~\eqref{eq:support-L-lower} give
\begin{align}
    \E{L(\bs{\rho}_j(\tau_j))} &\geq (1-\epsilon_j)\log N_j-O(\log\log N_j). \label{eq:proof-L-length}
\end{align}

For $j\in\mc{J}_{\mr{E}}$, we have $A_j=\log\frac{1}{\epsilon_j}$. The function $z\mapsto(1-z)\log\frac{1}{z}$ is convex and decreasing on $(0,1)$. Therefore, Jensen's inequality, the bound $\E{\epsilon_{\mr{MAP},j}}\leq\epsilon_j$, and~\eqref{eq:support-L-lower} give
\begin{align}
    \E{L(\bs{\rho}_j(\tau_j))}
    &\geq \E{(1-\epsilon_{\mr{MAP},j})\log\frac{1}{\epsilon_{\mr{MAP},j}}} \\
    &\geq \left(1-\E{\epsilon_{\mr{MAP},j}}\right)\log\frac{1}{\E{\epsilon_{\mr{MAP},j}}} \\
    &\geq (1-\epsilon_j)\log\frac{1}{\epsilon_j}=(1-\epsilon_j)A_j. \label{eq:proof-L-error}
\end{align}
The estimate in~\eqref{eq:proof-L-length} applies when $\log\frac{1}{\epsilon_j}\leq\log N_j$, and the estimate in~\eqref{eq:proof-L-error} applies when $\log\frac{1}{\epsilon_j}>\log N_j$. Together, they give
\begin{align}
    \E{L(\bs{\rho}_j(\tau_j))} &\geq (1-\epsilon_j)A_j-o(A_j). \label{eq:proof-L-unified}
\end{align}

It remains to convert the terminal quantity in~\eqref{eq:proof-L-unified} into accumulated capacity deficit. We choose
\begin{align}
    r_j &\triangleq \left(A_j\log A_j\right)^{-\frac{1}{3}}. \label{eq:proof-r-j}
\end{align}
The threshold choice in~\eqref{eq:proof-r-j} has the two required properties
\begin{align}
    h_{\mr{b}}(r_j)+r_j\log(|\mc{X}|-1) &= o(1), \qquad \frac{1}{r_j}+\frac{1}{2\lambda_{\min}^{3}r_j^2}=o(A_j). \label{eq:proof-r-properties}
\end{align}
For all sufficiently large $j$,~\eqref{eq:proof-r-properties} verifies the condition in~\eqref{eq:support-r-condition}. Moreover, $\log\frac{M_j}{M_j-1}=o(A_j)$, and $1-\epsilon_j\geq1-\bar\epsilon>0$. Hence, the argument of the positive part in~\eqref{eq:support-mutual-information-difference} is positive for all sufficiently large $j$. Applying \lemref{lem:mutual-information-difference} with $r=r_j$ and using~\eqref{eq:proof-delta-lower},~\eqref{eq:proof-L-unified}, and~\eqref{eq:proof-r-properties} give
\begin{align}
    \Delta_j &\geq \left(\frac{C}{C_1}+o(1)\right)\left((1-\epsilon_j)A_j-o(A_j)\right) = \frac{C}{C_1}(1-\epsilon_j)A_j-o(A_j). \label{eq:proof-delta-final-lower}
\end{align}
The lower bound in~\eqref{eq:proof-delta-final-lower} and the upper bound in~\eqref{eq:proof-delta-upper} are incompatible because $d(1-\bar\epsilon)A_j$ cannot be absorbed into the $o(A_j)$ remainder. This contradiction proves~\eqref{eq:main-converse}. Finally, $\log A_N=o(A_N)$ and $\frac{h_{\mr{b}}(\epsilon_N)}{1-\epsilon_N}=O(1)=o(A_N)$, so~\eqref{eq:main-achievability} and~\eqref{eq:main-converse} give~\eqref{eq:main-expansion}.
\hfill\IEEEQED

\section{Proof of \thmref{thm:structure}} \label{sec:proof-structure}

First-order optimality gives the conclusions about terminal branches and decoding times, while second-order optimality gives the communication and confirmation structure.

The Fano-capacity bound concerns only the total expected decoding time. The following lemma lower-bounds the expected decoding time on the correct-decoding event.

\begin{lemma}\label{lem:correct-decoding-time}
For every DMC with $C>0$, there exist constants $K_{\mr{dec}}>0$ and $M_0\geq2$, depending only on the channel, such that every VLF code with stopping time $\tau$, message-set size $M\geq M_0$, and average error probability $\Pe$ satisfies
\begin{align}
    \E{\tau1\{\widehat W=W\}} &\geq \frac{(1-\Pe)\log M}{C}-K_{\mr{dec}}\sqrt{\log M}. \label{eq:structure-correct-time}
\end{align}
\end{lemma}

\begin{IEEEproof}
See \appref{app:correct-decoding}.
\end{IEEEproof}

First, we prove part (a) of \thmref{thm:structure}. Recall from~\eqref{eq:main-sequence-events} that $Z_j=1-\max_{m\in[M_j]}\rho_{j,m}(\tau_j)$ is the terminal conditional MAP error, $\mc{E}_j=\{\widehat W_j\neq W_j\}$ is the error event, and $\mc{A}_j=\{Z_j>\frac{1}{2}\}$ is the event on which this conditional error exceeds $\frac{1}{2}$. For the $j$th code, let $\mathsf O_{\tau_j}$ be the padded terminal observation in~\eqref{eq:problem-terminal-observation} and set $\mc{G}_j\triangleq\sigma(\mathsf O_{\tau_j})$. By \lemref{lem:padded-observation} and the MAP-decoder definition in~\eqref{eq:main-sequence-map-decoder}, $\widehat W_j$ is $\mc{G}_j$-measurable and
\begin{align}
    Z_j &= \Prob{\mc{E}_j\mid\mc{G}_j} \quad\text{almost surely}. \label{eq:structure-conditional-error}
\end{align}
In particular, $\E{Z_j}=P_{{\mr{e}},j}\leq\epsilon$. Applying the Fano-capacity bound in~\eqref{eq:support-fano} with $\Pe=P_{{\mr{e}},j}$, dividing by $N_j$, and using~\eqref{eq:main-first-order-codes} give
\begin{align}
    P_{{\mr{e}},j} &\to \epsilon. \label{eq:structure-error-limit}
\end{align}

The first-order condition in~\eqref{eq:main-first-order-codes} also gives
\begin{align}
    \Delta_{\mr{F}}(N_j,M_j,\epsilon) &= o(N_j). \label{eq:structure-delta-order}
\end{align}
Since $\log(M_j-1)=\Theta(N_j)$, we have $M_j\geq3$ and $\epsilon<1-\frac{1}{M_j}$ for all sufficiently large $j$. Hence, \lemref{lem:terminal-concentration} and~\eqref{eq:structure-delta-order} imply
\begin{align}
    \E{Z_j(1-Z_j)} &\to 0. \label{eq:structure-terminal-product}
\end{align}
For every $z\in[0,1]$, we have
\begin{align}
    \left|1\left\{z>\frac{1}{2}\right\}-z\right| &\leq 2z(1-z). \label{eq:structure-endpoint-elementary}
\end{align}
Taking expectations in~\eqref{eq:structure-endpoint-elementary} with $z=Z_j$ and using $\E{Z_j}=P_{{\mr{e}},j}$ give $|\Prob{\mc{A}_j}-P_{{\mr{e}},j}|\leq2\E{Z_j(1-Z_j)}$. Hence,~\eqref{eq:structure-error-limit} and~\eqref{eq:structure-terminal-product} imply $\Prob{\mc{A}_j}\to\epsilon$. On $\mc{A}_j$, we have $Z_j>\frac{1}{2}$ and hence $1-Z_j\leq2Z_j(1-Z_j)$. On $\mc{A}_j^{\mr{c}}$, we have $Z_j\leq\frac{1}{2}$ and hence $Z_j\leq2Z_j(1-Z_j)$. Therefore, the conditional-error identity in~\eqref{eq:structure-conditional-error}, these pointwise bounds, and~\eqref{eq:structure-terminal-product} give
\begin{align}
    \Prob{\mc{A}_j\mathbin{\triangle}\mc{E}_j} &= \E{1\{\mc{A}_j\}(1-Z_j)+1\{\mc{A}_j^{\mr{c}}\}Z_j}\leq2\E{Z_j(1-Z_j)}\to0. \label{eq:structure-event-equivalence}
\end{align}
The same two pointwise bounds and $\Prob{\mc{A}_j}\to\epsilon\in(0,1)$ prove~\eqref{eq:main-branch-errors}.

We next derive the decoding-time conclusions. Applying \lemref{lem:correct-decoding-time} with $\Pe=P_{{\mr{e}},j}$ and using~\eqref{eq:main-first-order-codes} and~\eqref{eq:structure-error-limit} give
\begin{align}
    \E{\tau_j1\{\widehat W_j=W_j\}} &\geq N_j-o(N_j). \label{eq:structure-correct-time-lower}
\end{align}
Since $\E{\tau_j}\leq N_j$, it follows that
\begin{align}
    \E{\tau_j} &= N_j-o(N_j), \qquad \E{\tau_j1\{\mc{E}_j\}}=o(N_j). \label{eq:structure-length-and-error-time}
\end{align}
The measurability conclusion in \lemref{lem:padded-observation} also applies to $\tau_j$. Consequently,~\eqref{eq:structure-conditional-error} yields
\begin{align}
    \E{\tau_j Z_j} &= \E{\tau_j\E{1\{\mc{E}_j\}\mid\mc{G}_j}}=\E{\tau_j1\{\mc{E}_j\}}. \label{eq:structure-conditional-error-time}
\end{align}
Since $1\{\mc{A}_j\}\leq2Z_j$, equations~\eqref{eq:structure-length-and-error-time} and~\eqref{eq:structure-conditional-error-time} imply
\begin{align}
    \E{\tau_j1\{\mc{A}_j\}} &= o(N_j), \qquad \E{\tau_j1\{\mc{A}_j^{\mr{c}}\}}=N_j-o(N_j). \label{eq:structure-branch-time-result}
\end{align}
For every $\nu>0$, conditional Markov inequalities give
\begin{align}
    \Prob{\tau_j>\nu N_j\mid\mc{E}_j} &\leq \frac{\E{\tau_j1\{\mc{E}_j\}}}{\nu N_jP_{{\mr{e}},j}}, \qquad \Prob{\tau_j>\nu N_j\mid\mc{A}_j}\leq \frac{\E{\tau_j1\{\mc{A}_j\}}}{\nu N_j\Prob{\mc{A}_j}}. \label{eq:structure-tail-bounds}
\end{align}
The numerators are $o(N_j)$, and both conditioning probabilities converge to $\epsilon>0$. This completes the proof of part (a).

We now prove part (b) of \thmref{thm:structure} and assume the second-order expansion in~\eqref{eq:main-second-order-codes}. For the $j$th code, set $\Delta_j\triangleq\Delta_{\mr{F}}(N_j,M_j,\epsilon)$. Define the stopped capacity deficit by
\begin{align}
    S_j &\triangleq \E{\sum_{t=1}^{\infty}1\{\tau_j\geq t\}(C-I_{j,t})}. \label{eq:structure-capacity-deficit}
\end{align}
Assumption~\eqref{eq:main-second-order-codes} gives
\begin{align}
    \Delta_j &= (1-\epsilon)\frac{C}{C_1}\log N_j+o(\log N_j). \label{eq:structure-second-order-deficit-value}
\end{align}
The exact decomposition in~\eqref{eq:support-delta-decomposition} and the non-negativity established there imply $S_j\leq\Delta_j$ and $C(N_j-\E{\tau_j})\leq\Delta_j-S_j$.

\lemref{lem:terminal-concentration} and~\eqref{eq:structure-second-order-deficit-value} yield
\begin{align}
    \E{Z_j(1-Z_j)} &= O\left(\frac{\log N_j}{N_j}\right). \label{eq:structure-second-order-terminal-product}
\end{align}
Since $\E{Z_j}=P_{{\mr{e}},j}\leq\epsilon$, common support gives $\E{Z_j(1-Z_j)}>0$, and~\eqref{eq:structure-second-order-terminal-product} implies $\E{Z_j(1-Z_j)}<\frac{1-\epsilon}{4}$ for all sufficiently large $j$. Thus, \lemref{lem:scalar} applies with $Z=Z_j$. Combining that lemma with~\eqref{eq:support-L-lower} gives
\begin{align}
    \E{L(\bs{\rho}_j(\tau_j))} &\geq (1-\epsilon)\log N_j-O(\log\log N_j). \label{eq:structure-second-order-L-lower}
\end{align}
For fixed $\epsilon$, the threshold $r_j$ in the theorem is the specialization of~\eqref{eq:proof-r-j} to $A_j=\log N_j$. Hence,~\eqref{eq:proof-r-properties} makes it admissible in \lemref{lem:mutual-information-difference}; the additive $r_j$-dependent terms are $o(\log N_j)$, and the coefficient multiplying the positive part is $\frac{C}{C_1}-o(1)$. Moreover,~\eqref{eq:structure-second-order-L-lower},~\eqref{eq:proof-r-properties}, and $\log\frac{M_j}{M_j-1}=o(1)$ show that the argument of the positive part in~\eqref{eq:support-mutual-information-difference} is positive. Applying the lemma with $r=r_j$ therefore gives
\begin{align}
    S_j &\geq (1-\epsilon)\frac{C}{C_1}\log N_j-o(\log N_j). \label{eq:structure-capacity-deficit-lower}
\end{align}
Combining this lower bound with $S_j\leq\Delta_j$ and~\eqref{eq:structure-second-order-deficit-value} yields
\begin{align}
    S_j &= (1-\epsilon)\frac{C}{C_1}\log N_j+o(\log N_j), \label{eq:structure-S-value}\\
    0\leq C(N_j-\E{\tau_j}) &\leq \Delta_j-S_j=o(\log N_j), \\
    N_j-\E{\tau_j} &= o(\log N_j). \label{eq:structure-capacity-deficit-value}
\end{align}

To separate the low-posterior times $\mc{H}_{j,t}^{\mr{c}}$ from the high-posterior times $\mc{H}_{j,t}$, define the non-negative drift gap
\begin{align}
    G_{j,t} &\triangleq \operatorname{EJS}_{j,t}-I_{j,t}. \label{eq:structure-drift-gap}
\end{align}
Its non-negativity follows from~\eqref{eq:support-ejs-dominates-information}.
Define the expected stopped drift gaps in $\mc{H}_{j,t}^{\mr{c}}$ and $\mc{H}_{j,t}$ by
\begin{align}
    G_{j,\mr{com}} &\triangleq \E{\sum_{t=1}^{\infty}1\{\tau_j\geq t,\mc{H}_{j,t}^{\mr{c}}\}G_{j,t}}, \label{eq:structure-communication-drift-gap-definition}\\
    G_{j,\mr{conf}} &\triangleq \E{\sum_{t=1}^{\infty}1\{\tau_j\geq t,\mc{H}_{j,t}\}G_{j,t}}. \label{eq:structure-confirmation-drift-gap-definition}
\end{align}
The low-posterior estimate in~\eqref{eq:support-no-large-posterior-bound} and the telescoping bound~\eqref{eq:support-stopped-posterior-square} give
\begin{align}
    0\leq G_{j,\mr{com}} &\leq \frac{1}{r_j}+\frac{1}{2\lambda_{\min}^3r_j^2}=o(\log N_j). \label{eq:structure-communication-drift-gap}
\end{align}
Set $\vartheta_j\triangleq h_{\mr{b}}(r_j)+r_j\log(|\mc{X}|-1)=o(1)$. Applying the truncation and limiting argument in~\eqref{eq:support-stopped-L}--\eqref{eq:support-stopped-comparison}, while retaining the low- and high-posterior contributions separately, gives
\begin{align}
    \E{L(\bs{\rho}_j(\tau_j))}-L(\bs{\rho}_j(0)) &\leq G_{j,\mr{com}}+G_{j,\mr{conf}} \leq \frac{1}{r_j}+\frac{1}{2\lambda_{\min}^3r_j^2}+\frac{C_1}{C-\vartheta_j}S_j. \label{eq:structure-drift-gap-sandwich}
\end{align}
The leftmost expression in~\eqref{eq:structure-drift-gap-sandwich} is at least $(1-\epsilon)\log N_j-o(\log N_j)$ by~\eqref{eq:structure-second-order-L-lower} and $L(\bs{\rho}_j(0))=\log\frac{M_j}{M_j-1}=o(1)$ by~\eqref{eq:support-initial-L}. The final upper bound in~\eqref{eq:structure-drift-gap-sandwich} equals $(1-\epsilon)\log N_j+o(\log N_j)$ by~\eqref{eq:proof-r-properties} and~\eqref{eq:structure-S-value}. Hence,
\begin{align}
    G_{j,\mr{com}}+G_{j,\mr{conf}} &= (1-\epsilon)\log N_j+o(\log N_j). \label{eq:structure-total-drift-gap}
\end{align}
Combining~\eqref{eq:structure-total-drift-gap} with~\eqref{eq:structure-communication-drift-gap} gives
\begin{align}
    G_{j,\mr{conf}} &= (1-\epsilon)\log N_j+o(\log N_j). \label{eq:structure-confirmation-drift-gap}
\end{align}

On $\mc{H}_{j,t}$,~\eqref{eq:support-large-posterior-mutual-information} gives $I_{j,t}\leq\vartheta_j$, while~\eqref{eq:support-ejs-C-one} bounds the EJS divergence by $C_1$. Hence,
\begin{align}
    \frac{G_{j,\mr{conf}}}{C_1} &\leq \E{T_{j,\mr{conf}}} \leq \frac{S_j}{C-\vartheta_j}. \label{eq:structure-confirmation-occupancy-bounds}
\end{align}
Combining~\eqref{eq:structure-confirmation-occupancy-bounds},~\eqref{eq:structure-S-value}, and~\eqref{eq:structure-confirmation-drift-gap} proves the occupancy formula in~\eqref{eq:main-second-order-occupancies}. Moreover,
\begin{align}
    C_1\E{T_{j,\mr{conf}}}-G_{j,\mr{conf}} &= \E{\sum_{t=1}^{\infty}1\{\tau_j\geq t\}1\{\mc{H}_{j,t}\}\left(C_1-\operatorname{EJS}_{j,t}+I_{j,t}\right)}=o(\log N_j). \label{eq:structure-confirmation-deficiencies}
\end{align}
Since both $C_1-\operatorname{EJS}_{j,t}$ and $I_{j,t}$ are non-negative,~\eqref{eq:structure-confirmation-deficiencies} proves~\eqref{eq:main-confirmation-efficiency} and~\eqref{eq:main-confirmation-mutual-information}. The high-posterior contribution to $S_j$ equals $C\E{T_{j,\mr{conf}}}-\E{\sum_{t=1}^{\infty}1\{\tau_j\geq t,\mc{H}_{j,t}\}I_{j,t}}$. Equations~\eqref{eq:main-second-order-occupancies} and~\eqref{eq:main-confirmation-mutual-information} therefore show that this contribution is $(1-\epsilon)\frac{C}{C_1}\log N_j+o(\log N_j)$. Subtracting it from the total stopped capacity deficit in~\eqref{eq:structure-S-value} proves~\eqref{eq:main-communication-efficiency} and completes the proof.
\hfill\IEEEQED

\section{Proof of \thmref{thm:bec}} \label{sec:proof-bec}

We first establish a prefix-tree inequality, then prove the converse and the matching achievability.

For a binary prefix tree with $s$ leaves at depths $d_1,\ldots,d_s$, its total external path length is $\sum_{i=1}^{s}d_i$. Let $\ell_{\Sigma}(s)$ denote the minimum of this sum over all binary prefix trees with $s$ leaves. Since $\ell_{\mr{H}}(s)$ in~\eqref{eq:main-huffman-length} is the corresponding minimum average depth for $s$ equiprobable leaves,
\begin{align}
    \ell_{\Sigma}(s) &= s\ell_{\mr{H}}(s). \label{eq:bec-total-path}
\end{align}
For $s\geq2$, let $b\triangleq\lfloor\log_2 s\rfloor$. Then, $\ell_{\Sigma}(s)=s(b+2)-2^{b+1}$, and
\begin{align}
    \frac{\ell_{\Sigma}(s+1)}{s}-\frac{\ell_{\Sigma}(s)}{s-1}
    &= \frac{2^{b+1}-(b+2)}{s(s-1)}>0. \label{eq:bec-path-monotonicity}
\end{align}
At the endpoint $s+1=2^{b+1}$, the adjacent Huffman formulas agree because $\ell_{\Sigma}(2^{b+1})=2^{b+1}(b+1)$. Hence,~\eqref{eq:bec-path-monotonicity} holds for every $s\geq2$.

For the converse, realize the BEC using independent indicators $J_1,J_2,\ldots\sim\operatorname{Bernoulli}(1-\delta)$, independent of $(U,W)$, by setting $Y_t=X_t$ when $J_t=1$ and $Y_t=\mathsf e$ when $J_t=0$. Because $\{\tau=t\}\in\sigma(U,Y^t)$ and the alphabets are finite, the indicator of $\{\tau=t\}$ can be represented as a function of $(U,Y^t)$. Fix one such representation for every $t$ and, on an arbitrary output history, define the canonical stopping time as the first indicated time, with value $\infty$ if no time is indicated. This rule agrees with $\tau$ almost surely under the physical law.

Let $\Xi\triangleq(U,J_1,J_2,\ldots)$. For each realization $\upsilon$ of $\Xi$ and each $w\in[M]$, recursively evaluate the causal encoder and the canonical stopping rule to obtain the counterfactual path for message $w$. Let $\tau_w(\upsilon)$ denote the canonical stopping time on this counterfactual path. Since $W$ is uniform and independent of $\Xi$,
\begin{align}
    \E{\tau} &= \frac{1}{M}\sum_{w=1}^{M}\E{\tau_w(\Xi)}. \label{eq:bec-counterfactual-mean}
\end{align}
Thus, $\E{\tau_w(\Xi)}<\infty$ for every $w$. Because the message set is finite, there is a probability-one set of realizations $\upsilon$ on which all $M$ counterfactual paths stop in finite time. We perform the following deterministic prefix-tree argument on that set and then average over $\Xi$.

Once the erased channel uses are deleted, each message determines a finite binary terminal transcript. The distinct terminal transcripts are prefix-free. For a fixed realization of $\Xi$, the erasure locations and unerased channel-use indices are common to all counterfactual messages. Equal unerased prefixes therefore give equal full received histories through the earlier stopping time, including every erasure symbol. Thus, if one terminal transcript were a strict prefix of another, the common stopping rule would stop for both messages at the earlier stopping time. Messages that generate the same terminal transcript also have the same decoder output, so at most one of them is decoded correctly.

Let $K$ denote the number of distinct terminal transcripts generated by the $M$ counterfactual messages, so $K$ is a function of $\Xi$. Fix a realization $\upsilon$ in the probability-one set above. At most $K$ of the $M$ messages are decoded correctly, so the pathwise success fraction is at most $\frac{K}{M}$. Let $d_i$ be the length of transcript $i$, and let $s_i$ be the number of messages that produce it. Then, $\sum_{i=1}^{K}s_i=M$. If we append an optimal binary prefix tree with $s_i$ leaves to transcript $i$, we obtain a binary prefix tree with $M$ leaves. Therefore,
\begin{align}
    \sum_{i=1}^{K}\left(s_id_i+\ell_{\Sigma}(s_i)\right) &\geq \ell_{\Sigma}(M). \label{eq:bec-appended-tree}
\end{align}
The monotonicity in~\eqref{eq:bec-path-monotonicity}, with $\ell_{\Sigma}(1)=0$, gives
\begin{align}
    \ell_{\Sigma}(s_i) &\leq \frac{s_i-1}{M-1}\ell_{\Sigma}(M). \label{eq:bec-small-tree-bound}
\end{align}
Substituting~\eqref{eq:bec-small-tree-bound} into~\eqref{eq:bec-appended-tree} gives
\begin{align}
    \frac{1}{M}\sum_{i=1}^{K}s_id_i &\geq \frac{K-1}{M-1}\ell_{\mr{H}}(M). \label{eq:bec-pathwise-length}
\end{align}

Let $S_\tau$ denote the number of unerased outputs observed before stopping. For fixed $\Xi$, the left-hand side in~\eqref{eq:bec-pathwise-length} equals $\E{S_\tau\mid\Xi}$ because exactly $s_i$ messages produce a transcript of length $d_i$. Averaging~\eqref{eq:bec-pathwise-length} over $\Xi$ gives
\begin{align}
    \E{S_\tau} &\geq \frac{\E{K}-1}{M-1}\ell_{\mr{H}}(M) 
    \geq \frac{M(1-\Pe)-1}{M-1}\ell_{\mr{H}}(M), \label{eq:bec-success-length}
\end{align}
where the second inequality follows because the pathwise success fraction is at most $\frac{K}{M}$ and hence $M(1-\Pe)\leq\E{K}$.
With the unerased-output indicators $J_t$ introduced at the start of the converse,
\begin{align}
    S_\tau &= \sum_{t=1}^{\infty}1\{\tau\geq t\}J_t. \label{eq:bec-stopped-count}
\end{align}
The event $\{\tau\geq t\}$ is determined before channel use $t$, while $J_t$ is independent of the preceding history and has mean $1-\delta$. Tonelli's theorem gives
\begin{align}
    \E{S_\tau} &= (1-\delta)\E{\tau}. \label{eq:bec-mean-count}
\end{align}
Equations~\eqref{eq:bec-success-length} and~\eqref{eq:bec-mean-count}, together with $\Pe\leq\epsilon$, prove the converse in~\eqref{eq:main-bec}.

For achievability, we set
\begin{align}
    q &\triangleq \frac{M(1-\epsilon)-1}{M-1}. \label{eq:bec-mixture-probability}
\end{align}
Using common randomness, the code selects the immediate-guess branch with probability $1-q$ and the Huffman-transmission branch with probability $q$. In the first branch, the decoder stops at time zero and makes a fixed guess. In the second branch, the encoder transmits an optimal Huffman codeword and repeats each bit until it is received without erasure. The second branch has zero error and expected decoding time $\frac{\ell_{\mr{H}}(M)}{1-\delta}$. The resulting error probability and expected decoding time are
\begin{align}
    (1-q)\left(1-\frac{1}{M}\right) &= \epsilon, \qquad q\frac{\ell_{\mr{H}}(M)}{1-\delta}=\frac{M(1-\epsilon)-1}{M-1}\frac{\ell_{\mr{H}}(M)}{1-\delta}. \label{eq:bec-achievability-values}
\end{align}
If $\epsilon\geq1-\frac{1}{M}$, immediate guessing satisfies the error constraint. The strict inequality in~\eqref{eq:bec-path-monotonicity} also shows that equality in the pathwise prefix-tree bound requires either one terminal transcript or $M$ distinct terminal transcripts. These are the two cases used by the immediate-guess and Huffman-transmission branches.
\hfill\IEEEQED

\section{Conclusion} \label{sec:conclusion}

We close the order-$\log N$ converse gap left by the 2011 bounds of Polyanskiy \emph{et al.}~\cite{PolyanskiyPoorVerdu2011} for VLF coding in the non-vanishing error probability regime. For every DMC with $C>0$ and $C_1<\infty$, the result also determines the asymptotic VLF fundamental limit when the error probability vanishes subexponentially and when it decays exponentially at any rate strictly below $C_1$. Across these regimes, the fundamental limit has the expansion $\log \Mvlf(N,\epsilon_N)=\frac{NC}{1-\epsilon_N}-\frac{C}{C_1}A_N+o(A_N)$. The converse combines a stopped R\'enyi-entropy argument whose order approaches one with an application of the EJS divergence.

The structural theorem shows that every first-order-optimal code sequence with a non-vanishing error probability has an event of probability $\epsilon+o(1)$ that accounts for almost all decoding errors and has conditional mean decoding time $o(N)$; decoding is asymptotically reliable on the complementary event. Second-order optimality further requires communication and confirmation behavior. For the BEC, where $C_1=\infty$, the optimum instead interpolates exactly between immediate guessing and zero-error Huffman transmission with retransmission after erasures. The present results do not determine the second-order fundamental limit for VLSF codes or for VLF codes over general DMCs with $C_1=\infty$ beyond the BEC.

\appendices

\section{Proof of the Achievability Bound in \thmref{thm:main}} \label{app:achievability}

We derive~\eqref{eq:main-achievability} by specializing the modified Yamamoto--Itoh construction of Yavas and Tan to the present error-probability sequences.

Following the communication- and confirmation-phase construction in~\cite[Sec.~VI-A1]{YavasTan2025}, use the capacity-achieving input distribution $P_X^*$ associated with~\eqref{eq:notation-capacity-output} as the communication-phase input distribution, and choose $(x_{\mr{A}},x_{\mr{R}})$ such that
\begin{align}
    D\left(P_{Y|X=x_{\mr{A}}}\middle\|P_{Y|X=x_{\mr{R}}}\right) &= C_1, \qquad D_{\mr{R}}\triangleq D\left(P_{Y|X=x_{\mr{R}}}\middle\|P_{Y|X=x_{\mr{A}}}\right). \label{eq:app-ach-divergences}
\end{align}
Since $C_1>C>0$, the two conditional output distributions in the first divergence in~\eqref{eq:app-ach-divergences} are distinct. Since $C_1$ is the maximum over ordered input pairs, $C_1<\infty$ implies that both ordered divergences between these two distributions are finite. Hence, they have the same support and $0<D_{\mr{R}}<\infty$. The subscripts $\mr{A}$ and $\mr{R}$ denote the accept and reject symbols of the confirmation phase in the cited construction; the subscript $\mr{c}$ denotes its communication phase.

We next define the three overshoot constants in the cited bound. For a random variable $G$ with $\E{G}>0$, let $P_G$ denote its probability law and define
\begin{align}
    b(P_G) &\triangleq \min\left\{\frac{\E{(G^+)^2}}{\E{G}},\mathop{\mr{ess\,sup}}G\right\}. \label{eq:app-ach-b-function}
\end{align}
Under $P_X^*P_{Y|X}$, define
\begin{align}
    Z_{\mr{c}} &\triangleq \imath(X;Y), \label{eq:app-ach-Z-c}\\
    Z_{\mr{A}} &\triangleq \log\frac{P_{Y|X=x_{\mr{A}}}(Y_{\mr{A}})}{P_{Y|X=x_{\mr{R}}}(Y_{\mr{A}})}, \qquad Y_{\mr{A}}\sim P_{Y|X=x_{\mr{A}}}, \label{eq:app-ach-Z-A}\\
    Z_{\mr{R}} &\triangleq \log\frac{P_{Y|X=x_{\mr{R}}}(Y_{\mr{R}})}{P_{Y|X=x_{\mr{A}}}(Y_{\mr{R}})}, \qquad Y_{\mr{R}}\sim P_{Y|X=x_{\mr{R}}}, \label{eq:app-ach-Z-R}\\
    b_{\mr{c}} &\triangleq b(P_{Z_{\mr{c}}}), \qquad b_{\mr{A}}\triangleq b(P_{Z_{\mr{A}}}), \qquad b_{\mr{R}}\triangleq b(P_{Z_{\mr{R}}}). \label{eq:app-ach-b-constants}
\end{align}
The definitions give $\E{Z_{\mr{c}}}=C$, $\E{Z_{\mr{A}}}=C_1$, and $\E{Z_{\mr{R}}}=D_{\mr{R}}$. The three means are positive, so every argument of $b(\cdot)$ in~\eqref{eq:app-ach-b-constants} lies in the domain specified in~\eqref{eq:app-ach-b-function}. The constants $b_{\mr{c}}$, $b_{\mr{A}}$, and $b_{\mr{R}}$ are finite because the alphabets are finite and the conditional output distributions have common support.

The error and expected-length analyses in~\cite[Eqs.~(65) and~(82)--(84)]{YavasTan2025} give, for every integer $M\geq2$, $0<\gamma_1<\gamma_2$, $a_{\mr{A}},a_{\mr{R}}>0$, and $\epsilon_0\in(0,1)$, a VLF code whose expected decoding time is at most $(1-\epsilon_0)N'$ and whose error probability is at most $\epsilon_0+(1-\epsilon_0)\epsilon'$. With $\widetilde M\triangleq M-1$, the residual-error parameter $\epsilon'$ and the unrandomized duration parameter $N'$ that enter the error and mean-length bounds are
\begin{align}
    \epsilon' &\triangleq \widetilde M\left(\exp\{-(\gamma_1+a_{\mr{A}})\}+\exp\{-\gamma_2\}\right), \label{eq:app-ach-source-error}\\
    N' &\triangleq \frac{\gamma_1+b_{\mr{c}}}{C}+\left(\widetilde M\exp\{-\gamma_1\}+\exp\{-a_{\mr{R}}\}\right)\frac{\gamma_2-\gamma_1+b_{\mr{c}}}{C}+\frac{a_{\mr{A}}+b_{\mr{A}}}{C_1}+\widetilde M\exp\{-\gamma_1\}\frac{a_{\mr{R}}+b_{\mr{R}}}{D_{\mr{R}}}. \label{eq:app-ach-source-length}
\end{align}
Any common-randomness mixture in this construction can be reduced to at most three mass points while preserving its expected decoding time and error probability, by applying the Carath\'eodory argument in~\cite[Th.~19]{PolyanskiyPoorVerdu2011} to the mean-length and error-probability averages. Thus, the construction satisfies the common-randomness cardinality in \defnref{def:vlf-code}.

We now modify the parameter choices for fixed $\epsilon$ in~\cite[Eqs.~(85)--(91)]{YavasTan2025} so that the bound remains uniform over the permitted error-probability sequences. Set
\begin{align}
    s_N &\triangleq A_N+\log 2, \label{eq:app-ach-s-u}\\
    \gamma_1 &\triangleq \log\widetilde M+\log s_N, \qquad \gamma_2\triangleq \log\widetilde M+s_N+\log 2, \label{eq:app-ach-gammas}\\
    a_{\mr{A}} &\triangleq s_N+\log 2-\log s_N, \qquad a_{\mr{R}}\triangleq s_N. \label{eq:app-ach-thresholds}
\end{align}
For all sufficiently large $N$, these parameters are positive and satisfy $\gamma_1<\gamma_2$. The terms $\widetilde M\exp\{-(\gamma_1+a_{\mr{A}})\}$ and $\widetilde M\exp\{-\gamma_2\}$ in~\eqref{eq:app-ach-source-error} are equal, so the residual error $\epsilon_N'\triangleq\epsilon'$ satisfies
\begin{align}
    \epsilon_N' &= \exp\{-s_N\}. \label{eq:app-ach-continuation-error}
\end{align}
The same parameter choices give
\begin{align}
    \widetilde M\exp\{-\gamma_1\} &= \frac{1}{s_N}, \qquad \gamma_2-\gamma_1=s_N+\log 2-\log s_N. \label{eq:app-ach-cancellations}
\end{align}
The first and third terms on the right-hand side of~\eqref{eq:app-ach-source-length} combine as
\begin{align}
    \frac{\gamma_1+b_{\mr{c}}}{C}+\frac{a_{\mr{A}}+b_{\mr{A}}}{C_1}
    &= \frac{\log\widetilde M}{C}+\frac{s_N}{C_1}
    +\left(\frac{1}{C}-\frac{1}{C_1}\right)\log s_N
    +\frac{b_{\mr{c}}}{C}+\frac{\log 2+b_{\mr{A}}}{C_1}.
\end{align}
By~\eqref{eq:app-ach-cancellations}, the second and fourth terms on the right-hand side of~\eqref{eq:app-ach-source-length} are, respectively,
\begin{align}
    \left(\frac{1}{s_N}+\exp\{-s_N\}\right)\frac{s_N+\log 2-\log s_N+b_{\mr{c}}}{C} &= O(1), \label{eq:app-ach-bounded-second}\\
    \frac{1}{s_N}\frac{s_N+b_{\mr{R}}}{D_{\mr{R}}} &= O(1), \label{eq:app-ach-bounded-fourth}
\end{align}
uniformly over the permitted error-probability sequences. Since $s_N=A_N+O(1)$ and $\log s_N=\log A_N+O(1)$, there is a finite channel-dependent constant $K_{\mr{ch}}$ such that
\begin{align}
    N' &\leq \frac{\log\widetilde M}{C}+\frac{A_N}{C_1}+\left(\frac{1}{C}-\frac{1}{C_1}\right)\log A_N+K_{\mr{ch}}. \label{eq:app-ach-simplified-length}
\end{align}

It remains to choose the message-set size and the time-zero randomization. Let
\begin{align}
    T_N &\triangleq \frac{N(1-\epsilon_N')}{1-\epsilon_N}. \label{eq:app-ach-time-budget}
\end{align}
Since $A_N\geq\log\frac{1}{\epsilon_N}$ and $A_N\geq\log N$,
\begin{align}
    \epsilon_N' &\leq \frac{\epsilon_N}{2}, \qquad N\epsilon_N'\leq \frac{1}{2}. \label{eq:app-ach-residual-bounds}
\end{align}
We define
\begin{align}
    \Gamma_N &\triangleq CT_N-\frac{C}{C_1}A_N-\left(1-\frac{C}{C_1}\right)\log A_N-C(K_{\mr{ch}}+1). \label{eq:app-ach-Gamma-N}
\end{align}
Equation~\eqref{eq:app-ach-residual-bounds} gives $T_N\geq N$. Moreover,~\eqref{eq:main-reliability-condition} and $\log N=o(N)$ imply $A_N\leq(C_1-\eta)N$ for all sufficiently large $N$, while $\log A_N=O(\log N)$. Hence,
\begin{align}
    \Gamma_N &\geq \frac{C\eta}{C_1}N-O(\log N), \label{eq:app-ach-Gamma-lower}
\end{align}
and therefore $\Gamma_N\to\infty$.

We choose
\begin{align}
    M_N^{\mr{ach}} &\triangleq 1+\left\lfloor\exp\{\Gamma_N\}\right\rfloor. \label{eq:app-ach-integer-choice}
\end{align}
For all sufficiently large $N$, $\Gamma_N>0$ and hence $M_N^{\mr{ach}}\geq2$. We instantiate the construction in~\eqref{eq:app-ach-source-error}--\eqref{eq:app-ach-source-length} with
$\widetilde M=M_N^{\mr{ach}}-1=\lfloor\exp\{\Gamma_N\}\rfloor$ and $M=M_N^{\mr{ach}}$. Since $\log\widetilde M\leq\Gamma_N$, substituting~\eqref{eq:app-ach-Gamma-N} into~\eqref{eq:app-ach-simplified-length} gives
\begin{align}
    N' &\leq T_N-1<T_N. 
\end{align}
We set
\begin{align}
    \epsilon_0 &\triangleq \frac{\epsilon_N-\epsilon_N'}{1-\epsilon_N'}. \label{eq:app-ach-epsilon-zero}
\end{align}
Equation~\eqref{eq:app-ach-residual-bounds} gives $\epsilon_0\in(0,1)$. The mean-length and error-probability bounds in the cited construction satisfy
\begin{align}
    (1-\epsilon_0)N' &= \frac{1-\epsilon_N}{1-\epsilon_N'}N' \leq N, \qquad \epsilon_0+(1-\epsilon_0)\epsilon_N'=\epsilon_N. \label{eq:app-ach-final-constraints}
\end{align}
Consequently, an $(N,M_N^{\mr{ach}},\epsilon_N)$-VLF code exists, and therefore $M_N^{\mr{ach}}\leq\Mvlf(N,\epsilon_N)$.

The definition in~\eqref{eq:app-ach-integer-choice} also gives $M_N^{\mr{ach}}>\exp\{\Gamma_N\}$ and hence
\begin{align}
    \log M_N^{\mr{ach}} &> \Gamma_N. \label{eq:app-ach-rounding}
\end{align}
Moreover,
\begin{align}
    0 &\leq \frac{NC}{1-\epsilon_N}-CT_N = \frac{NC\epsilon_N'}{1-\epsilon_N} \leq \frac{C}{2(1-\bar\epsilon)}. \label{eq:app-ach-first-term-difference}
\end{align}
Combining~\eqref{eq:app-ach-Gamma-N},~\eqref{eq:app-ach-rounding}, and~\eqref{eq:app-ach-first-term-difference} gives
\begin{align}
    \log\Mvlf(N,\epsilon_N)
    &\geq \log M_N^{\mr{ach}} > \frac{NC}{1-\epsilon_N}
    -\frac{C}{C_1}A_N
    -\left(1-\frac{C}{C_1}\right)\log A_N
    -C(K_{\mr{ch}}+1)
    -\frac{C}{2(1-\bar\epsilon)}, 
\end{align}
which proves~\eqref{eq:main-achievability}.
\hfill\IEEEQED

\section{Proof of \lemref{lem:renyi}} \label{app:renyi}

We first prove~\eqref{eq:support-renyi-stopped} from a one-step entropy bound. Fix $t\in\mathbb N$, condition on $\mc{F}_{t-1}$, write $\bs{\rho}=\bs{\rho}(t-1)$, and define
\begin{align}
    B_\alpha &\triangleq \sum_{m=1}^{M}\rho_m^\alpha, \label{eq:app-renyi-B-alpha}\\
    P_{X,t}^{(\alpha)}(x) &\triangleq \frac{ \sum_{m:x_t(m)=x}\rho_m^\alpha }{ B_\alpha }, \qquad x\in\mc{X}. \label{eq:app-renyi-auxiliary-input}
\end{align}
Since $B_\alpha>0$, $P_{X,t}^{(\alpha)}$ is a probability distribution on $\mc{X}$. We restrict the following calculations to outputs $y$ satisfying $\overline P_t(y)>0$; conditional on $\mc{F}_{t-1}$, every omitted output symbol has probability zero. On this support, Bayes' rule gives
\begin{align}
    \rho_m(t) &= \frac{\rho_mP_{m,t}(Y_t)}{\overline P_t(Y_t)}, \qquad m\in[M]. \label{eq:app-renyi-next-posterior}
\end{align}
Consequently, the sum of the $\alpha$th powers of the updated posterior probabilities, evaluated at $Y_t=y$, is
\begin{align}
    \sum_{m=1}^{M}\left(\frac{\rho_mP_{m,t}(y)}{\overline P_t(y)}\right)^\alpha &= B_\alpha \frac{ \sum_{x\in\mc{X}}P_{X,t}^{(\alpha)}(x)P_{Y|X}(y|x)^\alpha }{ \overline P_t(y)^\alpha }. \label{eq:app-renyi-posterior-power}
\end{align}
We substitute~\eqref{eq:app-renyi-posterior-power} into the definition of R\'enyi entropy and obtain
\begin{align}
    H_\alpha(\bs{\rho}(t-1))-\E{H_\alpha(\bs{\rho}(t))\mid\mc{F}_{t-1}} &= \frac{1}{\alpha-1}\sum_{y:\overline P_t(y)>0}\overline P_t(y)\log\frac{\sum_{x\in\mc{X}}P_{X,t}^{(\alpha)}(x)P_{Y|X}(y|x)^\alpha}{\overline P_t(y)^\alpha}. \label{eq:app-renyi-exact-drift}
\end{align}
For this calculation, we define
\begin{align}
    z_{\alpha,t}(y) &\triangleq \left( \sum_{x\in\mc{X}}P_{X,t}^{(\alpha)}(x)P_{Y|X}(y|x)^\alpha \right)^{\frac{1}{\alpha}}, \label{eq:app-renyi-z}\\
    \nu_{\alpha,t} &\triangleq \sum_{y\in\mc{Y}}z_{\alpha,t}(y), \label{eq:app-renyi-nu}\\
    Q_{\alpha,t}(y) &\triangleq \frac{z_{\alpha,t}(y)}{\nu_{\alpha,t}}. \label{eq:app-renyi-Q}
\end{align}
If $\overline P_t(y)>0$, then $z_{\alpha,t}(y)>0$. Hence, $\overline P_t\ll Q_{\alpha,t}$. By~\eqref{eq:app-renyi-z}--\eqref{eq:app-renyi-Q}, the right-hand side of~\eqref{eq:app-renyi-exact-drift} equals the first expression in~\eqref{eq:app-renyi-one-step}. The non-negativity of relative entropy then gives
\begin{align}
    \frac{\alpha}{\alpha-1}\left[\log\nu_{\alpha,t}-D\left(\overline P_t\middle\|Q_{\alpha,t}\right)\right] &\leq \frac{\alpha}{\alpha-1}\log\nu_{\alpha,t}=I_\alpha^{\mr{S}}(P_{X,t}^{(\alpha)},P_{Y|X})\leq C_\alpha^{\mr{S}}. \label{eq:app-renyi-one-step}
\end{align}
Therefore, $H_\alpha(\bs{\rho}(t))+tC_\alpha^{\mr{S}}$ is a submartingale. Since the message prior is uniform, $H_\alpha(\bs{\rho}(0))=\log M$. Doob's optional stopping theorem~\cite[Sec.~10.10]{Williams1991} applied at $\min\{\tau,k\}$ gives
\begin{align}
    \log M &\leq \E{H_\alpha(\bs{\rho}(\min\{\tau,k\}))} +C_\alpha^{\mr{S}}\E{\min\{\tau,k\}}. \label{eq:app-renyi-bounded-stopping}
\end{align}
Because $\E{\tau}\leq N$, we have $\tau<\infty$ almost surely, and therefore $\bs{\rho}(\min\{\tau,k\})\to\bs{\rho}(\tau)$ almost surely. Since R\'enyi entropy belongs to $[0,\log M]$, dominated convergence applies to the first expectation, while monotone convergence applies to the second. Letting $k\to\infty$ gives~\eqref{eq:support-renyi-stopped}.

It remains to bound $C_{1+s}^{\mr{S}}$. For $\widetilde\alpha>1$ and probability distributions $P$ and $Q$ such that $P$ is absolutely continuous with respect to $Q$, the R\'enyi divergence of order $\widetilde\alpha$ is
\begin{align}
    D_{\widetilde\alpha}(P\|Q) &\triangleq \frac{1}{\widetilde\alpha-1} \log \sum_{y:Q(y)>0}P(y)^{\widetilde\alpha}Q(y)^{1-\widetilde\alpha}. \label{eq:app-renyi-divergence}
\end{align}
We set $D_{\widetilde\alpha}(P\|Q)=\infty$ when $P$ is not absolutely continuous with respect to $Q$. Evaluating the minimizing-center characterization in~\cite[Th.~2.2]{Sibson1969} at $Q_Y=P_Y^*$ gives the first inequality below, while bounding the weighted exponential average by its maximum gives the second.
\begin{align}
    I_{1+s}^{\mr{S}}(P_X,P_{Y|X}) &\leq \frac{1}{s}\log\sum_{x\in\mc{X}}P_X(x)\exp\left\{sD_{1+s}\left(P_{Y|X=x}\middle\|P_Y^*\right)\right\} \leq \max_{x\in\mc{X}}D_{1+s}\left(P_{Y|X=x}\middle\|P_Y^*\right). \label{eq:app-renyi-radius}
\end{align}
Equation~\eqref{eq:notation-capacity-condition} implies $P_{Y|X=x}\ll P_Y^*$ and, for $Y\sim P_{Y|X=x}$,
$\E{\imath(x;Y)}=D(P_{Y|X=x}\|P_Y^*)\leq C$ for every $x\in\mc{X}$. Hence, the R\'enyi divergences in~\eqref{eq:app-renyi-radius} are finite; output symbols outside the support of $P_Y^*$ may be omitted. The range of $y\mapsto\imath(x;y)$ has length at most $b_{\mr{ch}}$. Therefore, Hoeffding's lemma~\cite[Lem.~1]{Hoeffding1963} gives
\begin{align}
    D_{1+s}\left(P_{Y|X=x}\middle\|P_Y^*\right) &= \frac{1}{s}\log\E{\exp\{s\imath(x;Y)\}} \leq C+\frac{s b_{\mr{ch}}^2}{8}. \label{eq:app-renyi-hoeffding}
\end{align}
Taking the maximum over $P_X$ in~\eqref{eq:app-renyi-radius}, using the definition in~\eqref{eq:support-sibson-capacity}, and applying~\eqref{eq:app-renyi-hoeffding} proves~\eqref{eq:support-sibson-capacity-bound}.
\hfill\IEEEQED

\section{Proof of \lemref{lem:terminal-renyi}} \label{app:terminal-renyi}

Recall from \lemref{lem:terminal-renyi} that $\theta=1-\max_{m\in[M]}\rho_m$.
If $\theta=0$, then $H_{\alpha_M}(\bs{\rho})=0$, and~\eqref{eq:support-terminal-renyi} follows directly. We henceforth assume that $0<\theta\leq1-\frac{1}{M}$.

Choose a coordinate whose mass equals $1-\theta$. The remaining $M-1$ coordinates have total mass $\theta$. By convexity of $a\mapsto a^{\alpha_M}$, their contribution to the power sum is at least $\theta^{\alpha_M}(M-1)^{1-\alpha_M}=\exp\{-2\}\theta^{\alpha_M}$. Since $\frac{1}{1-\alpha_M}=-\frac{\log(M-1)}{2}<0$, substituting this lower bound into the definition of R\'enyi entropy gives
\begin{align}
    H_{\alpha_M}(\bs{\rho})
    &= -\frac{\log(M-1)}{2}\log\sum_{m=1}^{M}\rho_m^{\alpha_M} \\
    &\leq -\frac{\log(M-1)}{2} \log \left[ (1-\theta)^{1+\frac{2}{\log(M-1)}} +\exp\{-2\}\theta^{1+\frac{2}{\log(M-1)}} \right]. \label{eq:app-terminal-extremal}
\end{align}
We define
\begin{align}
    \omega(\theta) &\triangleq 1-\theta+\exp\{-2\}\theta, \label{eq:app-terminal-omega}\\
    w_0(\theta) &\triangleq \frac{1-\theta}{\omega(\theta)}, \qquad w_1(\theta)\triangleq \frac{\exp\{-2\}\theta}{\omega(\theta)}. \label{eq:app-terminal-weights}
\end{align}
The weights in~\eqref{eq:app-terminal-weights} are positive and sum to one. Factoring the argument of the logarithm in~\eqref{eq:app-terminal-extremal} as $\omega(\theta)$ times the left-hand side below, the weighted arithmetic-geometric mean inequality gives
\begin{align}
    w_0(\theta)(1-\theta)^{\frac{2}{\log(M-1)}}+w_1(\theta)\theta^{\frac{2}{\log(M-1)}} &\geq (1-\theta)^{\frac{2w_0(\theta)}{\log(M-1)}}\theta^{\frac{2w_1(\theta)}{\log(M-1)}}. \label{eq:app-terminal-am-gm}
\end{align}
We substitute~\eqref{eq:app-terminal-am-gm} into~\eqref{eq:app-terminal-extremal} and obtain
\begin{align}
    H_{\alpha_M}(\bs{\rho}) &\leq \log(M-1)\psi(\theta)+\zeta(\theta), \label{eq:app-terminal-psi-zeta}\\
    \psi(\theta) &\triangleq -\frac{1}{2}\log\omega(\theta), \label{eq:app-terminal-psi}\\
    \zeta(\theta) &\triangleq w_0(\theta)\log\frac{1}{1-\theta} +w_1(\theta)\log\frac{1}{\theta}. \label{eq:app-terminal-zeta}
\end{align}
For every $u\in(0,1]$, $u\log\frac{1}{u}\leq\exp\{-1\}$. Since $\omega(\theta)\geq\exp\{-2\}$,
\begin{align}
    w_0(\theta)\log\frac{1}{1-\theta} &\leq \exp\{2\}(1-\theta)\log\frac{1}{1-\theta} \leq \exp\{1\}, \label{eq:app-terminal-zeta-zero}\\
    w_1(\theta)\log\frac{1}{\theta} &\leq \theta\log\frac{1}{\theta} \leq \exp\{-1\}. \label{eq:app-terminal-zeta-one}
\end{align}
Thus, $\zeta(\theta)\leq c_0$.

The function $\psi$ is convex, with $\psi(0)=0$ and $\psi(1)=1$. Hence, $\psi(\theta)\leq\theta$ on $[0,1]$, so $\theta-\psi(\theta)$ is concave, non-negative, and zero at both endpoints. Applying concavity separately on $[0,\frac{1}{2}]$ and $[\frac{1}{2},1]$ gives
\begin{align}
    \theta-\psi(\theta)
    &\geq 2\left(\frac{1}{2}-\psi\left(\frac{1}{2}\right)\right)\min\{\theta,1-\theta\} \\
    &= \kappa\min\{\theta,1-\theta\}
    \geq \kappa\theta(1-\theta). \label{eq:app-terminal-concavity}
\end{align}
Combining $\psi(\theta)\leq\theta-\kappa\theta(1-\theta)$ with $\zeta(\theta)\leq c_0$ in~\eqref{eq:app-terminal-psi-zeta} proves~\eqref{eq:support-terminal-renyi}.
\hfill\IEEEQED

\section{Proof of \lemref{lem:scalar}} \label{app:scalar}

We define
\begin{align}
    \mc{B}_{\mr{lo}} &\triangleq \left\{Z\leq\frac{1}{2}\right\}, \qquad p_{\mr{lo}}\triangleq \Prob{\mc{B}_{\mr{lo}}}, \qquad \mu_{\mr{lo}}\triangleq \E{Z\,1\{\mc{B}_{\mr{lo}}\}}. \label{eq:app-scalar-definitions}
\end{align}
If $\Prob{Z=0}>0$, then the left-hand side in~\eqref{eq:support-scalar} is infinite and the result follows. We therefore assume that $Z>0$ almost surely.
On $\mc{B}_{\mr{lo}}$, $Z(1-Z)\geq\frac{Z}{2}$, while on $\mc{B}_{\mr{lo}}^{\mr{c}}$, $Z(1-Z)\geq\frac{1-Z}{2}$. Hence,
\begin{align}
    \mu_{\mr{lo}} &\leq 2v_Z, \label{eq:app-scalar-mu-bound}\\
    \E{(1-Z)\,1\{\mc{B}_{\mr{lo}}^{\mr{c}}\}} &\leq 2v_Z. \label{eq:app-scalar-complement-bound}
\end{align}
The identity $\E{Z}=\mu_{\mr{lo}}+(1-p_{\mr{lo}})-\E{(1-Z)1\{\mc{B}_{\mr{lo}}^{\mr{c}}\}}$, together with $\E{Z}\leq\epsilon$,~\eqref{eq:app-scalar-complement-bound}, and $\mu_{\mr{lo}}\geq0$, gives
\begin{align}
    p_{\mr{lo}} &\geq 1-\epsilon-2v_Z. \label{eq:app-scalar-probability-bound}
\end{align}
The assumption $v_Z<\frac{1-\epsilon}{4}$ makes the right-hand side in~\eqref{eq:app-scalar-probability-bound} positive, and hence $p_{\mr{lo}}>0$. Together with $Z>0$ almost surely, this implies $\mu_{\mr{lo}}>0$. Conditional Jensen's inequality gives
\begin{align}
    \E{1\{\mc{B}_{\mr{lo}}\}\log\frac{1}{Z}} &\geq p_{\mr{lo}}\log\frac{p_{\mr{lo}}}{\mu_{\mr{lo}}}. \label{eq:app-scalar-jensen}
\end{align}
Subtracting $\E{Z\,1\{\mc{B}_{\mr{lo}}\}\log\frac{1}{Z}}$ from both sides of~\eqref{eq:app-scalar-jensen} and using $z\log\frac{1}{z}\leq\exp\{-1\}$ for $z\in[0,1]$ give
\begin{align}
    \E{1\{\mc{B}_{\mr{lo}}\}(1-Z)\log\frac{1}{Z}} &\geq p_{\mr{lo}}\log\frac{p_{\mr{lo}}}{\mu_{\mr{lo}}} -\exp\{-1\}. \label{eq:app-scalar-subtraction}
\end{align}
To justify replacing $p_{\mr{lo}}$ by its lower bound $1-\epsilon-2v_Z$ and $\mu_{\mr{lo}}$ by its upper bound $2v_Z$, observe that, for $p>\mu>0$,
\begin{align}
    \frac{\partial}{\partial p}\left[p\log\frac{p}{\mu}\right] &= \log\frac{p}{\mu}+1>0, \qquad
    \frac{\partial}{\partial\mu}\left[p\log\frac{p}{\mu}\right]= -\frac{p}{\mu}<0. \label{eq:app-scalar-monotonicity}
\end{align}
The assumption $v_Z<\frac{1-\epsilon}{4}$ and~\eqref{eq:app-scalar-mu-bound}--\eqref{eq:app-scalar-probability-bound} give
\begin{align}
    p_{\mr{lo}} &\geq1-\epsilon-2v_Z>2v_Z\geq\mu_{\mr{lo}}. \label{eq:app-scalar-order}
\end{align}
Thus,~\eqref{eq:app-scalar-monotonicity} permits the substitutions in~\eqref{eq:app-scalar-subtraction}. Adding the non-negative complementary contribution $\E{1\{\mc{B}_{\mr{lo}}^{\mr{c}}\}(1-Z)\log\frac{1}{Z}}$ to its left-hand side proves~\eqref{eq:support-scalar}.
\hfill\IEEEQED

\section{Proof of \lemref{lem:correct-decoding-time}} \label{app:correct-decoding}

We combine a fixed-time change-of-measure estimate with the layer-cake representation of the expected decoding time.

To handle the dependence induced by feedback, we condition on the encoder's full history and define
\begin{align}
    \mc{K}_{t-1} &\triangleq \sigma(W,U,Y^{t-1}). \label{eq:app-correct-full-history}
\end{align}
The input $X_t$ is $\mc{K}_{t-1}$-measurable. We define the conditional mean
\begin{align}
    \mu_t &\triangleq \E{\imath(X_t;Y_t)\mid\mc{K}_{t-1}}. \label{eq:app-correct-increment}
\end{align}
Memorylessness and~\eqref{eq:notation-capacity-condition} give
\begin{align}
    \mu_t &= D\left(P_{Y|X=X_t}\middle\|P_Y^*\right) \leq C. \label{eq:app-correct-mean}
\end{align}
Moreover, $\imath(X_t;Y_t)-\mu_t$ has conditional range at most $b_{\mr{ch}}$. The process
\begin{align}
    \ms{M}_n &\triangleq \sum_{t=1}^{n}\left(\imath(X_t;Y_t)-\mu_t\right) \label{eq:app-correct-martingale}
\end{align}
is a martingale. We also define
\begin{align}
    \Lambda_n &\triangleq \sum_{t=1}^{n}\imath(X_t;Y_t). \label{eq:app-correct-sum}
\end{align}
Since $\Lambda_n-nC\leq \ms{M}_n$, the Azuma--Hoeffding inequality gives, when $b_{\mr{ch}}>0$,
\begin{align}
    \Prob{\Lambda_n-nC\geq a} &\leq \exp\left\{-\frac{2a^2}{n b_{\mr{ch}}^2}\right\}, \qquad a>0. \label{eq:app-correct-hoeffding}
\end{align}
If $b_{\mr{ch}}=0$, then $\imath(X_t;Y_t)=\mu_t\leq C$ almost surely and
\begin{align}
    \Lambda_n &\leq nC \quad\text{almost surely}. \label{eq:app-correct-zero-range}
\end{align}

We next relate early correct decoding to $\Lambda_n$. Fix finite-history indicator versions of the stopping events $\{\tau=t\}$, and on arbitrary output sequences let the canonical stopping rule be the first indicated time, with value $\infty$ if no time is indicated. At a finite canonical stopping time $t$, use the decision $\ms{g}_t(U,Y^t)$. This stopping rule and decision agree with $(\tau,\widehat W)$ almost surely under the physical law. Fix $n\in\mathbb N_0$ and extend the encoder after $\tau$ by arbitrary causal inputs. Define the $\mc{F}_n$-measurable decision
\begin{align}
    \widehat W^{(n)} &\triangleq \begin{cases}\widehat W,&\tau\leq n,\\1,&\tau>n.\end{cases} \label{eq:app-correct-fixed-time-decision}
\end{align}
Let $P_n$ be the physical law of $(W,U,Y^n)$, and let $Q_n$ retain the law of $(W,U)$ but generate $Y_1,\ldots,Y_n$ independently according to $P_Y^*$. Under $P_n$ and $Q_n$, the encoder, canonical stopping rule, and truncated decoder use the same measurable maps of the canonical coordinates. For every $\xi>0$,
\begin{align}
    \Prob{\widehat W=W,\tau\leq n}=\Prob{\widehat W^{(n)}=W,\tau\leq n} &\leq \Prob{\Lambda_n\geq\log M-\xi}+\exp\{-\xi\}. \label{eq:app-correct-change-measure}
\end{align}
To prove~\eqref{eq:app-correct-change-measure}, observe that under $Q_n$, the message is independent of $(U,Y^n)$ and remains uniform. Thus,
\begin{align}
    Q_n[\widehat W^{(n)}=W,\tau\leq n] &\leq \frac{1}{M}. \label{eq:app-correct-auxiliary-success}
\end{align}
Define the likelihood ratio
\begin{align}
    L_n &\triangleq \prod_{t=1}^{n}\begin{cases}
        \frac{P_{Y|X}(Y_t|X_t)}{P_Y^*(Y_t)},&P_Y^*(Y_t)>0,\\
        0,&P_Y^*(Y_t)=0,
    \end{cases} \label{eq:app-correct-product-likelihood}
\end{align}
with $L_0\triangleq1$. Conditionally on $(W,U,Y^{t-1})$, the $t$th output law is $P_{Y|X=X_t}$ under $P_n$ and $P_Y^*$ under $Q_n$. Multiplying these conditional likelihood ratios over $t\in[n]$ gives
\begin{align}
    L_n &= \frac{\mr{d} P_n}{\mr{d}Q_n}. \label{eq:app-correct-radon-nikodym}
\end{align}
The capacity condition~\eqref{eq:notation-capacity-condition} ensures that $P_{Y|X=x}\ll P_Y^*$ for every $x$. Under $P_n$, every realized numerator in~\eqref{eq:app-correct-product-likelihood} is positive, and hence
\begin{align}
    L_n &>0, \qquad \log L_n=\Lambda_n \quad P_n\text{-almost surely}. \label{eq:app-correct-log-likelihood}
\end{align}
We retain $L_n$, rather than writing $\exp\{\Lambda_n\}$ under $Q_n$, because $Q_n$ may generate histories containing a transition with zero physical probability; on each such history, $L_n=0$. On the event $\mc{B}_n\triangleq\{\widehat W^{(n)}=W,\tau\leq n\}$, change of measure and~\eqref{eq:app-correct-auxiliary-success} give
\begin{align}
    P_n[\mc{B}_n,L_n< M\exp\{-\xi\}] &= \mathbb E_{Q_n}\left[L_n1\{\mc{B}_n,L_n<M\exp\{-\xi\}\}\right] \leq M\exp\{-\xi\}Q_n[\mc{B}_n]\leq\exp\{-\xi\}. \label{eq:app-correct-low-likelihood}
\end{align}
Under $P_n$,~\eqref{eq:app-correct-log-likelihood} identifies $\{L_n\geq M\exp\{-\xi\}\}$ with $\{\Lambda_n\geq\log M-\xi\}$. Splitting $\mc{B}_n$ according to the events $\{L_n<M\exp\{-\xi\}\}$ and $\{L_n\geq M\exp\{-\xi\}\}$ and applying~\eqref{eq:app-correct-low-likelihood} proves~\eqref{eq:app-correct-change-measure}.

We set
\begin{align}
    n_M &\triangleq \left\lfloor\frac{\log M}{C}\right\rfloor, \qquad \xi_M\triangleq 2\log\log M. \label{eq:app-correct-cutoff}
\end{align}
We next show that the fixed-time probabilities appearing in~\eqref{eq:app-correct-change-measure} have the aggregate bound
\begin{align}
    \sum_{n=0}^{n_M-1}\Prob{\Lambda_n\geq\log M-\xi_M} &= O(\sqrt{\log M}), \label{eq:app-correct-tail-sum}
\end{align}
for all sufficiently large $M$. The $n=0$ summand is zero for all sufficiently large $M$. When $b_{\mr{ch}}>0$, we use~\eqref{eq:app-correct-hoeffding} for $n\geq1$ and write $n=n_M-1-j$. If $j\geq\frac{2\xi_M}{C}$, then
\begin{align}
    \log M-\xi_M-nC &\geq C(1+j)-\xi_M \geq \frac{Cj}{2}. \label{eq:app-correct-threshold-lower}
\end{align}
Since $n\leq\frac{\log M}{C}$, the Hoeffding exponent in~\eqref{eq:app-correct-hoeffding} is at least $\frac{C^3j^2}{2b_{\mr{ch}}^2\log M}$. The indices satisfying $j<\frac{2\xi_M}{C}$ contribute $O(\log\log M)$. For a channel-dependent constant $c>0$, the remaining terms are bounded by
\begin{align}
    \sum_{j=0}^{\infty}\exp\left\{-\frac{c j^2}{\log M}\right\} &= O(\sqrt{\log M}). \label{eq:app-correct-gaussian-sum}
\end{align}
When $b_{\mr{ch}}=0$,~\eqref{eq:app-correct-zero-range} makes every term with $j\geq\frac{2\xi_M}{C}$ equal to zero, while the remaining $O(\log\log M)$ indices are bounded trivially. Thus,~\eqref{eq:app-correct-tail-sum} holds in both cases. Finally,
\begin{align}
    n_M\exp\{-\xi_M\} &= O\left(\frac{1}{\log M}\right). \label{eq:app-correct-threshold-sum}
\end{align}

The layer-cake identity gives
\begin{align}
    \E{\tau1\{\widehat W=W\}} &= \sum_{n=0}^{\infty}\Prob{\widehat W=W,\tau>n} \geq n_M\Prob{\widehat W=W}-\sum_{n=0}^{n_M-1}\Prob{\widehat W=W,\tau\leq n}. \label{eq:app-correct-layer-cake}
\end{align}
Substituting~\eqref{eq:app-correct-change-measure},~\eqref{eq:app-correct-tail-sum}, and~\eqref{eq:app-correct-threshold-sum} into~\eqref{eq:app-correct-layer-cake}, using $\Prob{\widehat W=W}=1-\Pe$, and absorbing the bounded floor error prove~\eqref{eq:structure-correct-time}.
\hfill\IEEEQED

\section*{Acknowledgment}

\emph{Generative-AI use disclosure:}
This work grew out of the author's effort to close gaps left by his earlier research in finite-blocklength information theory. During the development and preparation of the manuscript, the author used OpenAI's ChatGPT with GPT-5.6 Sol (Ultra mode) and Anthropic Claude Fable 5 (Max effort) to discuss and critically examine possible proof techniques. In these exchanges, ChatGPT suggested extending Burnashev's one-step Shannon-entropy drift bound~\cite[Lem.~2]{Burnashev1976} to R\'enyi entropy. The author developed and verified this extension, which became a key ingredient in the converse argument that closes the logarithmic gap in the VLF fundamental limit. The two systems were also used to review proof steps, check algebraic and probabilistic calculations for possible errors, examine statements about prior work, and refine the exposition. The author independently verified every mathematical statement and proof in the final manuscript and checked the cited claims against their original sources. The author alone is responsible for the content of the manuscript, including all claims, interpretations, and any remaining errors.

\bibliographystyle{IEEEtran}
\bibliography{biblio2}

\end{document}